\documentclass[aps,pra,twocolumn,superscriptaddress,longbibliography]{revtex4-2}

\usepackage{mathdots}
\usepackage{mathrsfs}
\usepackage{amssymb,amsmath,amsthm}
\usepackage{graphicx}
\usepackage{bm,bbm}
\usepackage{csquotes}
\MakeOuterQuote{"}
\usepackage{tikz}
\usepackage{pifont}     
\usepackage{algorithm}
\usepackage{algcompatible}

\newcommand{\ket}[1]{\mbox{$ | #1 \rangle $}}
\newcommand{\bra}[1]{\mbox{$ \langle #1 | $}}

\newcommand{\tr}{\operatorname{tr}}

\newcommand{\diag}{\operatorname{diag}}

\newcommand{\supp}{\operatorname{supp}}

\newcommand{\caH}{\mathcal{H}}
\newcommand{\caW}{\mathcal{W}}
\newcommand{\caS}{\mathcal{S}}

\newcommand{\caZ}{\mathcal{Z}}
\newcommand{\bcaZ}{\bar{\mathcal{Z}}}
\newcommand{\rmd}{\mathrm{d}}
\newcommand{\rme}{\mathrm{e}}
\newcommand{\rmi}{\mathrm{i}}
\newcommand{\rmU}{\mathrm{U}}
\newcommand{\rmRe}{\mathrm{Re}}

\newcommand{\AS}{\mathrm{AS}}

\newcommand{\bfk}{\mathbf{k}}

\newcommand{\scrS}{\mathscr{S}}

 \newcommand{\bbC}{\mathbb{C}}

\newcommand{\1}{\operatorname{\uppercase\expandafter{\romannumeral1}}}
\newcommand{\2}{\operatorname{\uppercase\expandafter{\romannumeral2}}}
\newcommand{\3}{\operatorname{\uppercase\expandafter{\romannumeral3}}}
\newcommand{\4}{\operatorname{\uppercase\expandafter{\romannumeral4}}}
\newcommand{\5}{\operatorname{\uppercase\expandafter{\romannumeral5}}}
\newcommand{\6}{\operatorname{\uppercase\expandafter{\romannumeral6}}}
\newcommand{\7}{\operatorname{\uppercase\expandafter{\romannumeral7}}}
\newcommand{\8}{\operatorname{\uppercase\expandafter{\romannumeral8}}}
\newcommand{\9}{\operatorname{\uppercase\expandafter{\romannumeral9}}}

\newtheoremstyle{note}
  {\topsep/2}              	
  {\topsep/2}            	
  {}                        
  {\parindent}             	
  {\itshape}                
  {.}                       
  {5pt plus 1pt minus 1pt}  
  {}

\newtheorem{theorem}{Theorem}
\newtheorem{lemma}{Lemma}

\newtheorem{proposition}{Proposition}

\theoremstyle{definition}

\theoremstyle{remark}

\def\eqref#1{\textup{(\ref{#1})}}
\newcommand{\eref}[1]{Eq.~\textup{(\ref{#1})}}
\newcommand{\lref}[1]{Lemma~\ref{#1}}
\newcommand{\tref}[1]{Theorem~\ref{#1}}

\newcommand{\eqsref}[2]{Eqs.~(\ref{#1}) and (\ref{#2})}
\newcommand{\eqssref}[3]{Eqs.~(\ref{#1}), (\ref{#2}), and (\ref{#3})}

\def\<{\langle}  
\def\>{\rangle}  

\begin{document}
\title{Verification of phased Dicke states}

\author{Zihao Li}
\affiliation{State Key Laboratory of Surface Physics, Fudan University, Shanghai 200433, China}
\affiliation{Department of Physics and Center for Field Theory and Particle Physics, Fudan University, Shanghai 200433, China}
\affiliation{Institute for Nanoelectronic Devices and Quantum Computing, Fudan University, Shanghai 200433, China}

\author{Yun-Guang Han}
\affiliation{State Key Laboratory of Surface Physics, Fudan University, Shanghai 200433, China}
\affiliation{Department of Physics and Center for Field Theory and Particle Physics, Fudan University, Shanghai 200433, China}
\affiliation{Institute for Nanoelectronic Devices and Quantum Computing, Fudan University, Shanghai 200433, China}

\author{Hao-Feng Sun}
\affiliation{State Key Laboratory of Surface Physics, Fudan University, Shanghai 200433, China}
\affiliation{Department of Physics and Center for Field Theory and Particle Physics, Fudan University, Shanghai 200433, China}
\affiliation{Institute for Nanoelectronic Devices and Quantum Computing, Fudan University, Shanghai 200433, China}

\author{Jiangwei Shang}
\affiliation{Key Laboratory of Advanced Optoelectronic Quantum Architecture and Measurement of
Ministry of Education, School of Physics, Beijing Institute of Technology, Beijing 100081, China}
\affiliation{State Key Laboratory of Surface Physics, Fudan University, Shanghai 200433, China}

\author{Huangjun Zhu}
\email{zhuhuangjun@fudan.edu.cn}
\affiliation{State Key Laboratory of Surface Physics, Fudan University, Shanghai 200433, China}
\affiliation{Department of Physics and Center for Field Theory and Particle Physics, Fudan University, Shanghai 200433, China}
\affiliation{Institute for Nanoelectronic Devices and Quantum Computing, Fudan University, Shanghai 200433, China}
\affiliation{Collaborative Innovation Center of Advanced Microstructures, Nanjing 210093, China}

\begin{abstract}
Dicke states are typical examples of quantum states with genuine multipartite entanglement.
They are valuable resources in many quantum information processing tasks, including multiparty quantum communication and quantum metrology.
Phased Dicke states are a generalization of Dicke states and include antisymmetric basis states as a special example. These states are useful in atomic and molecular physics besides quantum information processing.
Here we propose practical and efficient protocols based on adaptive local projective measurements for verifying all phased Dicke states, including $W$ states and qudit Dicke states. To verify any $n$-partite phased Dicke state within infidelity $\epsilon$ and significance level $\delta$,
the number of tests required is only $O(n\epsilon^{-1}\ln\delta^{-1})$, which is linear in  $n$  and  is exponentially more efficient than traditional tomographic approaches.
In the case of $W$ states, the number of tests can be further reduced to  $O(\sqrt{n}\,\epsilon^{-1}\ln\delta^{-1})$.
Moreover, we construct an optimal  protocol for any antisymmetric basis state; the number of tests required decreases (rather than increases) monotonically with $n$. This
is the only optimal protocol known for
multipartite nonstabilizer states.
\end{abstract}

\date{\today}
\maketitle

\section{Introduction}
Quantum states with genuine multipartite entanglement (GME) play crucial roles in quantum information processing and foundational studies \cite{Horo09,Guhne09}.
Dicke states \cite{Dicke54,Wei03} are one of the most important multipartite quantum states other than stabilizer states. They are useful in many quantum information processing tasks, such as multiparty quantum communication
and quantum metrology \cite{Boure06,Kiesel07,Luo17,Murao99,Preve09,Wiec09,Pezze08}. Phased Dicke states are a generalization of Dicke states constructed by introducing phase changes and they are
equally useful in the above research areas \cite{Krammer09,Chiuri10}. Besides the usual Dicke states, antisymmetric basis states are a prominent example of phased Dicke states \cite{Denni01,Zanar02,Bravyi03}; they are usually used to represent the fermions,
and play a paramount role in atomic and molecular physics.
By now  numerous experiments have been performed to prepare and
engineer Dicke states \cite{Kiesel07,Preve09,Wiec09,Zou18,Cruz18,Haff05}, phased Dicke states
\cite{Chiuri10,Chiuri12}, and antisymmetric basis states \cite{Zhang16,Berry18} in various platforms.

In practice, it is usually extremely difficult to prepare quantum states with GME perfectly, and the success probability decays rapidly with  the number of particles.
Therefore, it is crucial to verify these states with high precision efficiently using limited resources. For the convenience of applications, it is also desirable to achieve this task using only local operations and classical communication (LOCC).
Unfortunately, traditional tomographic approaches are notoriously inefficient and are too resource consuming for systems with more than ten qubits \cite{Haff05}, since they extract too much unnecessary information. Although direct fidelity estimation \cite{FlamL11} can improve the efficiency significantly, it is still not satisfactory except for some special states, like stabilizer states.

Recently, an alternative approach known as quantum state verification \cite{HayaMT06,Haya09,Aolita15,Hang17,PLM18,ZhuEVQPSshort19,ZhuEVQPSlong19} has attracted increasing attention because of its potential to achieve a much higher efficiency. So far efficient verification protocols based on LOCC have been constructed for bipartite pure states \cite{LHZ19,Wang19,Yu19}, stabilizer states (including graph states) \cite{HayaM15,PLM18,ZhuH19E,ZhuEVQPSlong19},
hypergraph states \cite{ZhuH19E}, weighted graph states \cite{HayaTake19}, and qubit Dicke states \cite{Liu19}. Moreover, a similar idea can be applied  to verifying quantum gates and processes \cite{WuS19,LSYZ20,ZhuZ20,Zeng19}.
On the other hand, efficient protocols are still not available for many other important quantum states, including qudit Dicke states and phased Dicke states in particular. In addition, it is extremely difficult to construct optimal verification protocols, especially for nonstabilizer states. So far optimal protocols under LOCC are known only for maximally entangled states \cite{HayaMT06,Haya09,ZhuH19O}, two-qubit pure states \cite{Wang19},  Greenberger-Horne-Zeilinger (GHZ) states \cite{LiGHZ19}, and some other stabilizer states \cite{DangHZ20}. Any progress on this topic is of  interest to both theoretical studies and practical applications.

In this work, we construct highly efficient and practical verification protocols for all phased Dicke states, including  $W$ states and qudit Dicke states.
Our protocols only require adaptive local projective measurements with classical communication, which are as simple as one can expect.
Incidentally,   no efficient protocols based on nonadaptive measurements are known so far for verifying general bipartite pure states (cf.~Refs.~\cite{LHZ19,Wang19,Yu19}), not to mention multipartite states.
To verify any $n$-partite phased Dicke state within infidelity $\epsilon$ and significance level $\delta$,
the number of tests required  is only $O(n\epsilon^{-1}\ln\delta^{-1})$, which is linear in  $n$. So  our protocols can extract the key information---the fidelity with the target state---exponentially more efficiently than traditional  approaches, including tomography and direct fidelity estimation.
In the case of $W$ states, the number of tests can be further reduced to $O(\sqrt{n}\,\epsilon^{-1}\ln\delta^{-1})$, which  is quadratically fewer compared with the best verification protocol known in the literature \cite{Liu19}.
For the three-qubit $W$ state, one of our protocols is  almost optimal under LOCC; in addition, this protocol  is useful for fidelity estimation because the verification operator is homogeneous \cite{ZhuEVQPSlong19}. Moreover, we construct an optimal verification protocol for every antisymmetric basis state;  the number of tests required decreases  monotonically with $n$.
This is the only optimal protocol known so far for
multipartite nonstabilizer states.   For quantum states with GME, such optimal protocols were known previously only  for GHZ states \cite{LiGHZ19}, which are stabilizer states and have Schmidt decomposition (optimal protocols for some other stabilizer states were constructed recently \cite{DangHZ20} after the initial posting of this paper).
In the course of study, we  introduce several tools for improving the efficiency of a  given verification protocol, which are useful to studying quantum verification in general.

\section{Pure state verification}
\subsection{Basic framework}
Before presenting our main results, let us take a brief review on the basic framework of pure state verification \cite{PLM18,ZhuEVQPSshort19,ZhuEVQPSlong19}.
Suppose there is a quantum device that is expected to produce the pure target state $|\Psi\>\in \caH$.
However, some errors may occur when the device is working, and  it actually produces the states $\sigma_1, \sigma_2, \ldots, \sigma_N$  in $N$ runs.
Let $\epsilon_j:=1-\<\Psi|\sigma_j|\Psi\>$ denote the infidelity between $\sigma_j$ and the target state, and let $\bar{\epsilon}:=\sum_j\epsilon_j/N$ denote the average infidelity.
Our aim is to verify whether these states are sufficiently close to the target state on average, that is, whether the
average infidelity $\bar{\epsilon}$ is smaller than some threshold $\epsilon$.

To achieve this task, for each state $\sigma_j$ the verifier performs a test and accepts the states produced if and only if (iff) all tests are passed.
Each test is specified by a two-outcome measurement $\{E_l,\openone-E_l\}$, which is chosen randomly with probability $p_l$ from a set of accessible measurements.
Here the test operator $E_l$ corresponds to passing the test and satisfies the condition $E_l|\Psi\>=|\Psi\>$, so that the target state can always pass the test.
If the infidelity of $\sigma_j$ satisfies $\epsilon_j\geq\tilde\epsilon$ for some threshold $\tilde\epsilon\geq0$, then the maximum probability that $\sigma_j$ can pass each test on average is given by \cite{PLM18,ZhuEVQPSlong19}
\begin{equation}
\max_{\<\Psi|\sigma|\Psi\>\leq 1-\tilde\epsilon }\tr(\Omega \sigma)=1- [1-\lambda_2(\Omega)]\tilde\epsilon=1- \nu(\Omega)\tilde\epsilon,
\end{equation}
where $\Omega:=\sum_{l} p_l E_l$ is called the verification operator or a strategy, $\lambda_2(\Omega)$ denotes the second largest
eigenvalue of $\Omega$, and $\nu(\Omega):=1-\lambda_2(\Omega)$ is the spectral gap from the maximum eigenvalue.
The probability of passing all $N$ tests is at most $\prod_j[1- \nu(\Omega)\epsilon_j]\leq[1- \nu(\Omega)\bar{\epsilon}]^N$.
To ensure the condition $\bar{\epsilon}<\epsilon$ with significance level $\delta$, it suffices to take \cite{ZhuEVQPSshort19,ZhuEVQPSlong19}
\begin{equation}\label{eq:NumberTest}
N=\biggl\lceil\frac{ \ln \delta}{\ln[1-\nu(\Omega)\epsilon]}\biggr\rceil\approx \frac{ \ln \delta^{-1}}{\nu(\Omega)\epsilon},
\end{equation}
where the approximation  is applicable  when $\nu(\Omega)\epsilon\ll 1$.
According to this equation, the efficiency of a verification strategy $\Omega$ is mainly determined by its spectral gap $\nu(\Omega)$.

\subsection{\label{sec:Opt}Optimization of test probabilities}
To optimize the verification efficiency, we need to maximize the spectral gap of the verification operator for the target state $|\Psi\>$ or minimize the second largest eigenvalue. Suppose the set of test operators $\{E_l\}_l$ for $|\Psi\>$ is fixed, then we need to optimize the probabilities for  performing individual tests.  Given a general verification operator  of the form $\Omega=\sum_l p_l E_l$, the second largest eigenvalue of $\Omega$ reads
\begin{equation}
\lambda_2(\Omega)=\|\bar{\Omega}\|=\biggl\|\sum_l p_l \bar{E_l}\biggr\|,
\end{equation}
where $\|\cdot\|$ denotes the operator norm and
\begin{equation}
\bar{\Omega}:=\Omega-|\Psi\>\<\Psi|,\quad \bar{E}_l:=E_l-|\Psi\>\<\Psi|.
\end{equation}
Note that $\lambda_2(\Omega)$ is convex in $\{p_l\}_l$, so that  $\nu(\Omega)$ is concave in $\{p_l\}_l$. In addition,  the minimum of $\lambda_2(\Omega)$ over  $\{p_l\}_l$ can be computed via semidefinite programming,
\begin{equation}\label{eq:MinGapSDP}
\begin{aligned}
&\text{minimize} &f &  \\
&\text{subject to} &\sum_l p_l \bar{E_l}\leq f\openone, \quad p_l \geq0,& & \\
&&	\sum_l p_l=1.& \\
\end{aligned}
\end{equation}

The minimum in \eref{eq:MinGapSDP} can  be derived analytically when $\Omega$ consists of two projective tests thanks to the following lemma, which is proved in Appendix~\ref{app:LemmaTwoTestProof}.
\begin{lemma}\label{lem:2TestStrategy}
	Suppose $\Omega=p P_1 +(1-p)P_2$, where $0\leq p\leq 1$ and  $P_1, P_2$ are test projectors for $|\Psi\>$ with ranks at least 2. Then
	$\lambda_2(\Omega)\geq (1+\sqrt{q})/2$	and
	$\nu(\Omega)\leq(1-\sqrt{q})/2$, where
\begin{equation}
q:=\|\bar{P}_1\bar{P}_2 \bar{P}_1\|=\max_{|\phi\>\in \supp(\bar{P}_1)} \<\phi|P_2|\phi\>.
\end{equation}	
	If $q<1$, then the upper bound for $\nu(\Omega)$ is saturated iff $p=1/2$.
\end{lemma}
Note that any test projector based on LOCC has rank at least 2 if $|\Psi\>$ is entangled.
Previously, \lref{lem:2TestStrategy} was known in the special case in which $\bar{P}_1$ and $\bar{P}_2$ are orthogonal \cite{ZhuH19O}.

\subsection{\label{sec:sym}Symmetrization of verification operators}
Here we consider another recipe for improving the verification efficiency by employing the symmetry of the target state $|\Psi\>$; similar ideas have already found applications in verifying  bipartite pure states  \cite{PLM18,Wang19,Yu19}.
Suppose $\Omega$ is a verification operator for $|\Psi\>$, so that $\Omega\geq |\Psi\>\<\Psi|$ and $\Omega|\Psi\>=|\Psi\>$. Let $U$ be a unitary operator that leaves $|\Psi\>$ invariant up to a phase factor, that is, $U|\Psi\>\<\Psi|U^\dag =|\Psi\>\<\Psi|$ or, equivalently,
$U|\Psi\>=\rme^{\rmi \phi} |\Psi\>$, where $\phi$ is a phase (a real number). Then  $U\Omega U^\dag$ is also a valid verification operator for  $|\Psi\>$.  Moreover, $U\Omega U^\dag$ and $\Omega$ have the same spectral gap, that is,
\begin{equation}\label{eq:SpectralGapU}
\nu(U\Omega U^\dag)=\nu(\Omega).
\end{equation}

Let $G$ be the group generated by product unitaries and permutations that leave $|\Psi\>\<\Psi|$ invariant. Then $U\Omega U^\dag$ for any $U\in G$
can be realized by LOCC (is separable) iff $\Omega$ can be realized by LOCC (is separable).
Let $S$ be a subset of $G$ and define
\begin{equation}\label{eq:OmegaSym}
\Omega^S:=\int_SU\Omega U^\dag \rmd U,
\end{equation}
where the integral is taken with respect to the normalized probability measure induced from the Haar measure on $G$ (see Chapter~11 in Ref.~\cite{Halmos13} for example). The measure reduces to the normalized Haar measure on $S$ when $S$ is a group and reduces to the counting measure (see page~27 in Ref.~\cite{Schilling05} for example) when $S$ is finite. The verification operator $\Omega$ is called $S$-invariant if $\Omega^S=\Omega$.

If the verification strategy $\Omega$ consists of $m$ distinct tests and has the form $\Omega=\sum_l p_l E_l$, then
\begin{equation}
\Omega^S=\sum_l p_l E_l^S,
\end{equation}
where
\begin{equation}
E_l^S:=\int_SUE_l U^\dag \rmd U.
\end{equation}
If $S$ is a finite set with cardinality $|S|$, then the above equation reduces to
\begin{equation}
E_l^S=\frac{1}{|S|} \sum_{U\in S}U E_l U^\dag.
\end{equation}
If each test operator $E_l$ is a projector, then each $E_l^S$
can be realized by at most $|S|$ distinct projective tests. Therefore, $\Omega^S$ can be realized by at most $m|S|$ distinct projective tests.

\begin{proposition}\label{pro:GapSym}
	Suppose $S\subseteq H\leq G$. Then
\begin{equation}
\nu(\Omega)\leq\nu(\Omega^S)\leq\nu(\Omega^H)\leq \nu(\Omega^G).
\end{equation}	
\end{proposition}
Here the notation $S\subseteq H$ means $S$ is a subset of $H$;  by contrast, the notation
$H\leq G$ means $H$ is a subgroup of $G$. Proposition~\ref{pro:GapSym} shows that symmetrization is an effective way for improving the verification efficiency.
\begin{proof}
The inequality $\nu(\Omega)\leq\nu(\Omega^S)$
follows from \eref{eq:SpectralGapU} and the fact that $\nu(\Omega)$ is concave in $\Omega$.
The inequality $\nu(\Omega^S)\leq\nu(\Omega^H)$
follows from  the inequality $\nu(\Omega)\leq\nu(\Omega^S)$ and the fact that $\Omega^H=(\Omega^S)^H$ given that $S$ is a subset of the group $H$. The inequality $\nu(\Omega^H)\leq\nu(\Omega^G)$
follows from the inequality $\nu(\Omega^S)\leq\nu(\Omega^H)$.
\end{proof}

The following proposition is useful to reducing the number of distinct tests when constructing a verification strategy based on the symmetrization procedure. It is  a corollary of \eref{eq:OmegaSymProjGen} below.
\begin{proposition}\label{pro:GapSym3}
	Suppose $S\leq H\leq G$; in addition, $S$ and $H$ have the same number of irreducible components. Then $\Omega^{S}=\Omega^{H}$ and $\nu(\Omega^S)=\nu(\Omega^{H})$.
\end{proposition}

Suppose $S$ is a subgroup of $G$ and has $r$ inequivalent irreducible components with dimensions $d_j$ and multiplicities $m_j$, respectively (here we view $S$ as a representation of itself).  Then the Hilbert space $\caH$ decomposes into
\begin{equation}
\caH=\bigoplus_{j=1}^r \caH_j\otimes \bbC^{m_j},
\end{equation}
where $\caH_j$ has dimension $d_j$ and  carries the $j$th irreducible representation, and $\bbC^{m_j}$ denotes the multiplicity space.
Let $Q_j$ be the projector onto $\caH_j\otimes \bbC^{m_j}$, then
\begin{equation}\label{eq:OmegaSymProjGen}
\Omega^S=\sum_{j=1}^r \frac{1}{d_j} \bigl[\openone_{\caH_j}\otimes \tr_{\caH_j}(Q_j \Omega)\bigr] Q_j,
\end{equation}
where $\tr_{\caH_j}$ means the partial trace over $\caH_j$
(cf. the appendix of Ref.~\cite{Gross07}).
If all irreducible components of $S$ are inequivalent, that is, $m_j=1$ for $j=1,2,\ldots, r$, then \eref{eq:OmegaSymProjGen} reduces to
\begin{equation}\label{eq:OmegaSymProj}
\Omega^S=\sum_{j=1}^r \frac{\tr(Q_j \Omega)}{d_j} Q_j,
\end{equation}
where $Q_j$ is the projector onto the $j$th irreducible component. In this case, all $S$-invariant verification operators commute with each other.

\newcommand{\xratio}{1.2}%
\newcommand{\yratio}{1.5}%
\newcommand{\cmark}{\ding{51}}%
\newcommand{\xmark}{\ding{55}}%

\begin{figure*}
\begin{center}
	\includegraphics[width=15.54cm]{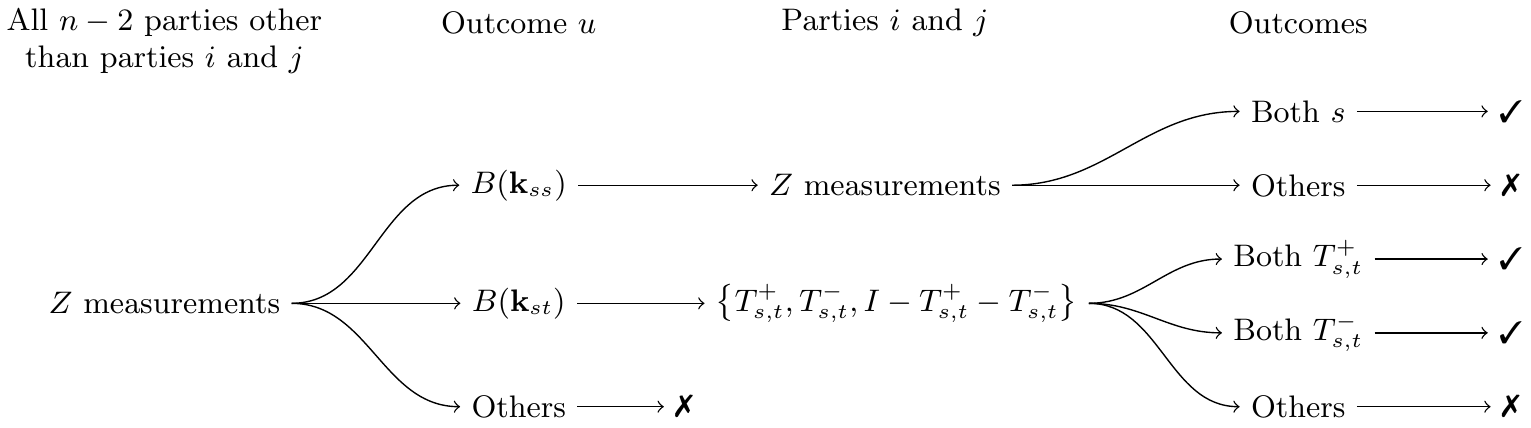}
	\caption{\label{fig:measurement-I} Schematic view of the test $P_{i,j}$ used to verify the Dicke state $\ket{D(\bfk)}$.
All $n-2$ parties other than parties $i$ and $j$ first perform the generalized Pauli-$Z$ measurement and send the outcome $u$ to parties $i$ and $j$.
Conditioned on this outcome, parties $i$ and $j$ then perform  suitable projective measurements. The outcomes corresponding to passing the test are marked by "\cmark".
}
\end{center}
\end{figure*}

\section{VERIFICATION OF Qudit Dicke STATES}
In this section we construct an efficient  protocol for verifying general qudit Dicke states \cite{Wei03}.
Previously,  efficient protocols were known only for qubit Dicke states \cite{Liu19}.

\subsection{Dicke states}
 Up to  a local unitary transformation, each $n$-qudit Dicke state can be labeled by  a sequence of $r+1$ ordered positive integers that sum up to $n$, where $r\leq d-1$. Let
\begin{equation}\label{eq:bfk}
\bfk:=(k_0,k_1,\dots,k_r),
\end{equation}
where $k_0,k_1,\dots,k_r$ are positive integers that satisfy the conditions $\sum_{j=0}^{r}k_j=n$ and  $k_0\geq k_1\geq\cdots\geq k_r\geq 1$.
Denote by $B(\bfk)$ the set of all sequences of $n$ symbols in which $k_i$ symbols are equal to $i$ for $i=0,1,\dots,r$.
Then the $n$-partite Dicke state corresponding to the sequence $\bfk$ has the form
\begin{equation}\label{eq:quditDstate}
\ket{D(\bfk)}=\frac{1}{\sqrt{m}} \sum_{u\in B(\bfk)}|u\>,
\end{equation}
where $m:=|B(\bfk)|=n!/\big(\prod _{j=0}^r k_j!\big)$.
It is worth pointing out that here we consider all Dicke states that can be defined for $n$-qudit systems with local dimension $d$, while some of these states can also  be defined for systems with smaller local dimensions. To avoid trivial cases, we assume that  $n\geq3$ and $k_0<n$ in the rest of this paper unless it is stated otherwise.


When  $\bfk=(2,1,1)$ for example,
the set  $B(\bfk)$ is composed of all sequences of four symbols in which two symbols are equal to 0, one symbol is equal to 1, and one symbol is equal to 2.
More concretely, $B(\bfk)=\{$0012, 0021, 0102, 0120, 0201, 0210, 1002, 1020, 1200, 2001, 2010, 2100$\}$.
The corresponding Dicke state reads
\begin{align}
\ket{D(\bfk)}= \ & \frac{1}{\sqrt{12}}\big( \ket{0012}+\ket{0021}+\ket{0102}+\ket{0120}\nonumber\\
&+\ket{0201}  +\ket{0210}+\ket{1002}+\ket{1020}+\ket{1200} \nonumber\\
&+\ket{2001}+\ket{2010}+\ket{2100}\big).
\end{align}

When $r=1$, $\ket{D(\bfk)}$ is a familiar qubit Dicke state. If in addition  $k_1=1$, then the  Dicke state reduces to a $W$ state \cite{Haff05},
\begin{equation}\label{eq:Wstate}
\ket{W_n}=\frac{1}{\sqrt{n}}\sum_{u\in B_n^1}|u\>,
\end{equation}
where  $B_n^1$ is  the set of strings in $\{0,1\}^n$ with Hamming weight 1.
In particular, the three-qubit $W$ state ($n=3$)  reads
\begin{equation}\label{eq:3qubitW}
\ket{W_3}=\frac{1}{\sqrt{3}}\big(\ket{001}+\ket{010}+\ket{100}\big).
\end{equation}

Denote by $S$ the group of all permutations of the $n$ parties (realized as unitary transformations); denote by $H$ the group of all unitary transformations of the form $U^{\otimes n}$, where $U$ is diagonal in the computational basis; let $G=SH=HS$. The Dicke state $\ket{D(\bfk)}$ is invariant under
any permutation of the $n$ parties and is thus invariant under the action of $S$. In addition, it is invariant (up to an overall phase factor) under any unitary transformation in $H$ or $G$. These observations are instructive to constructing efficient protocols for verifying the  state $\ket{D(\bfk)}$. Given any verification operator $\Omega$ for $\ket{D(\bfk)}$, we can construct potentially more efficient verification operators $\Omega^H$, $\Omega^S$, $\Omega^G$ according to \eref{eq:OmegaSym} and Proposition~\ref{pro:GapSym}.

\subsection{\label{sec:DickeVerify}Efficient verification of qudit Dicke states}
To construct an efficient protocol for verifying the Dicke state $|D(\bfk)\>$ defined in \eref{eq:quditDstate},
it is convenient to introduce some additional notations and concepts. Let
\begin{align}
\bfk_t^s &:=(k_0,\dots,k_s+1,  \dots,k_t-1,  \dots,k_r),   \label{eq:kst} \\
\bfk_{st}&:=(k_0,\dots,k_s-1,  \dots,k_t-1,  \dots,k_r),      \label{eq:kst2}   \\
\bfk_{ss}      &:=(k_0,\dots,k_s-2,\dots,k_r)  \quad  {\rm for}\ \, k_s\geq2. \label{eq:kss}
\end{align}
Here we assume that $0\leq s, t\leq r$ and $s\ne t$ in \eref{eq:kst},  $0\leq s< t\leq r$ in \eref{eq:kst2}, and
$0\leq s\leq r$ in \eref{eq:kss}.
Now the sets  $B(\bfk_t^s)$, $B(\bfk_{st})$, and $B(\bfk_{ss})$ can be defined in the same way as $B(\bfk)$.
The generalized Pauli-$Z$ operator acting on a single qudit is defined as
\begin{equation}\label{eq:Z}
Z=\sum_{j=0}^{d-1} \omega^j |j\>\<j|,  \qquad  \omega=\rme^{2\pi\rmi/d}.
\end{equation}
The generalized Pauli-$Z$ measurement is the projective measurement on the computational basis.

\begin{table}
	\caption{\label{tab:Protocol}
		The efficiencies of various verification strategies for the $n$-qubit $W$ state, three-qubit $W$ state, $n$-partite Dicke states, phased Dicke states, and antisymmetric basis state.
		Here $\nu(\Omega)$ denotes the spectral gap of each strategy, and $N(\epsilon,\delta,\Omega)$ denotes the number of tests required to verify the target state within infidelity $\epsilon$ and significance level $\delta$. In addition, the coefficients $a$ and $b$ read $a=\sqrt{\pi/2}\tanh(\pi/2)\approx 1.15$ and $b=\sqrt{\pi/2}\coth(\pi/2)\approx 1.37$.
	}		
	\begin{math}
	\begin{array}{c|cc}
	\hline\hline
	\mbox{Strategy} &\nu(\Omega) &N(\epsilon,\delta,\Omega) \\[0.5ex]
	\hline
	\Omega_{W_n}   \ (n\gg1, n \ $is odd$)\    & b/(4\sqrt{n})        &4\sqrt{n} (b\epsilon)^{-1}\ln \delta^{-1}               \\[0.5ex]
	\Omega_{W_n}   \ (n\gg1, n \ $is even$)    & a/(4\sqrt{n})        &4\sqrt{n} (a\epsilon)^{-1}\ln \delta^{-1}               \\[0.5ex]
	\Omega_{W_n}^G \ (n\gg1, n \ $is odd$)\    & b/\sqrt{n}         &\sqrt{n} (b\epsilon)^{-1}\ln \delta^{-1}              \\[0.5ex]
	\Omega_{W_n}^G \ (n\gg1, n \ $is even$)    & a/\sqrt{n}         &\sqrt{n} (a\epsilon)^{-1}\ln \delta^{-1}              \\[0.5ex]
	\Omega_{\1}\ (n=3)                                &  0.305                &3.28 \epsilon^{-1}\ln \delta^{-1}                       \\[0.5ex]
	\Omega_{\2}\  (n=3)                              & 5/8                   &(8/5) \epsilon^{-1}\ln \delta^{-1}                      \\[0.5ex]
	\Omega_\bfk \ $and$\ \Omega_\bfk^\phi \  (\bfk=(2,1))    & 1/3     &3 \epsilon^{-1}\ln \delta^{-1}                          \\[0.5ex]
	\Omega_\bfk \ $and$\ \Omega_\bfk^\phi \  (\bfk\ne(2,1))  & 1/(n-1) &(n-1) \epsilon^{-1}\ln \delta^{-1}                      \\[0.5ex]
	\Omega_{\AS_n\!}                           & 1/(n-1)               &(n-1) \epsilon^{-1}\ln \delta^{-1}                      \\[0.5ex]		
	\Omega^{\tilde{G}}_{\AS_n\!}               & n/(n+1)               &(n+1)n^{-1} \epsilon^{-1}\ln \delta^{-1}                \\[0.5ex]
	\hline\hline
	\end{array}	
	\end{math}
\end{table}

Our verification protocol consists of $\binom{n}{2}$ distinct tests performed with  uniform probabilities. Each test is associated with a pair of parties among the $n$ parties and is based on adaptive local projective measurements.
To be specific, the test $P_{i,j}$ associated with parties $i$ and $j$ is  illustrated in Fig.~\ref{fig:measurement-I} and  realized as follows.
All $n-2$ parties other than parties $i$ and $j$ perform  the generalized Pauli-$Z$ measurements, and their outcomes are labeled by a sequence $u$ of $n-2$ symbols, which corresponds to the product state $|u\>$. The measurements of parties $i$ and $j$ depend on the outcome $u$, and  we need to distinguish three cases. Suppose $k_0,k_1,\dots,k_g\geq2$ and $k_{g+1}=k_{g+2}=\cdots k_r=1$, where $-1\leq g\leq r$.
\begin{enumerate}
\item[1.] $u\in B(\bfk_{ss})$ with $0\leq s\leq g$. \\
In this case, the normalized reduced state of parties $i$ and $j$ reads $|s\>_i|s\>_j$ (if the target Dicke state is measured).
Then the two parties both perform $Z$ measurement, and the test is passed if they both obtain outcome $s$.

\item[2.] $u\in B(\bfk_{st})$ with $0\leq s<t\leq r$.\\
In this case, the normalized reduced state of parties $i$ and $j$ reads $\frac{1}{\sqrt{2}}(|s\>_i|t\>_j+|t\>_i|s\>_j)$.
Then the two parties both perform the projective measurement $\big\{T_{s,t}^+,T_{s,t}^-,I-T_{s,t}^+ -T_{s,t}^-\big\}$, where $I$ is the identity operator for one qudit and
\begin{align}
T_{s,t}^+=\frac{1}{2}(\ket{s}+\ket{t})(\bra{s}+\bra{t}),\label{eq:Pst+} \\
T_{s,t}^-=\frac{1}{2}(\ket{s}-\ket{t})(\bra{s}-\bra{t}).\label{eq:Pst-}
\end{align}
The test is passed if they both obtain the first outcome (corresponding to $T_{s,t}^+$) or if they both obtain the second outcome (corresponding to $T_{s,t}^-$).

\item[3.] Other cases.\\
The state cannot be the target state $\ket{D(\bfk)}$, so the test is not passed.
\end{enumerate}

\begin{figure}
\begin{center}
	\includegraphics[width=8.6cm]{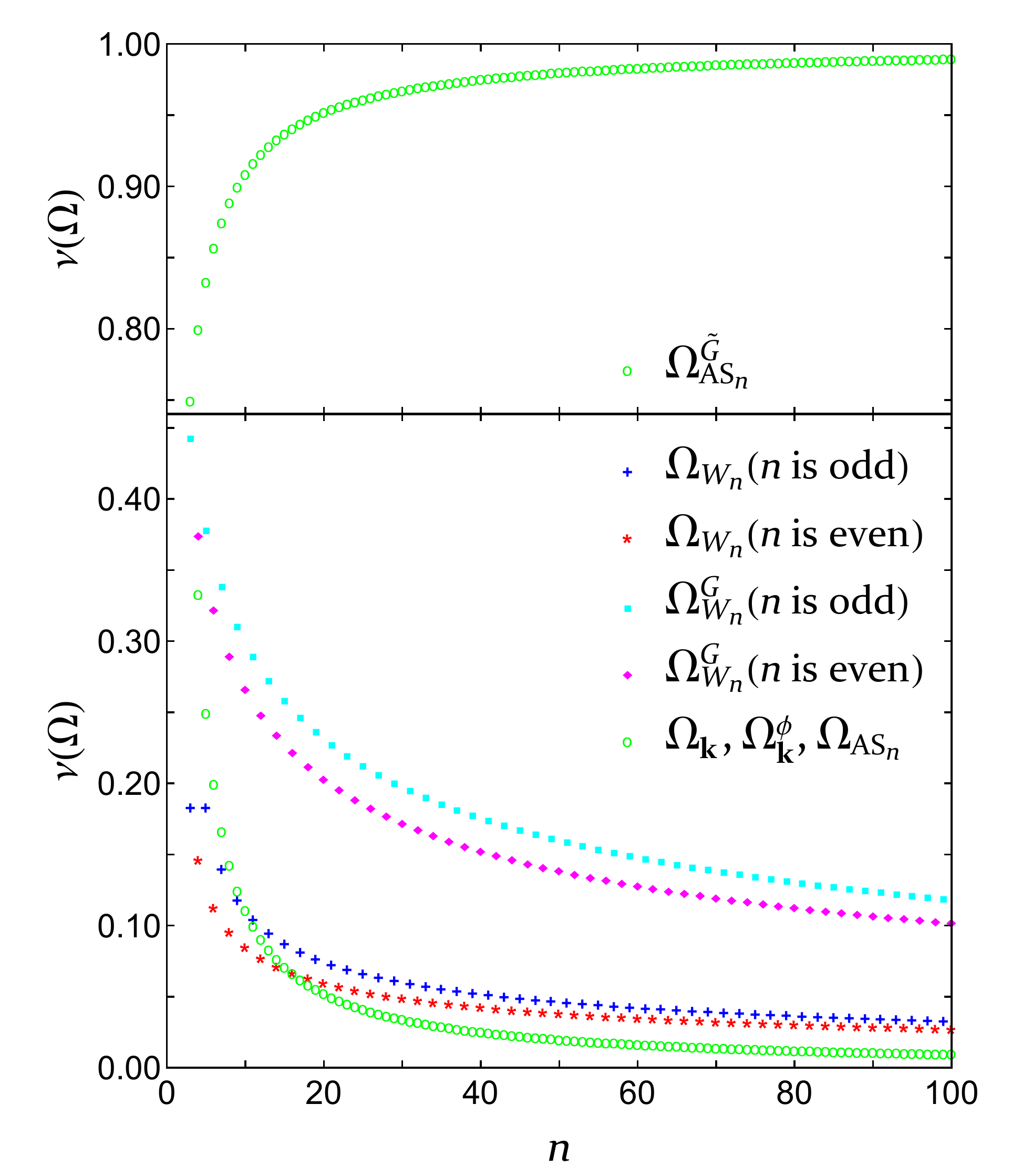}
	\caption{\label{fig:nu} Spectral gaps  $\nu(\Omega)$ of verification strategies for the $n$-qubit $W$ state, $n$-partite Dicke states, phased Dicke states, and antisymmetric basis state. The values of $\nu\big(\Omega_{W_n}\big)$ and $\nu\big(\Omega^G_{W_n}\big)$ oscillate with the parity of $n$.
Strategies $\Omega_\bfk$, $\Omega_\bfk^\phi$, and $\Omega_{\AS_n}$ have the same spectral gap when $n\geq4$  [cf.~\eqssref{eq:nuOmegaD}{eq:nuOmegaD'}{eq:nuOmegaAS}].
}
\end{center}
\end{figure}

The resulting test projector reads
\begin{align}\label{eq:Pij}
P_{i,j}=&   \sum_{s=0}^g \bcaZ_{i,j}(\bfk_{ss}) \otimes \big[(\ket{s}\bra{s})^{\otimes2}\big]_{i,j}\nonumber\\
&+\sum_{s<t}   \bcaZ_{i,j}(\bfk_{st}) \otimes \big[(T_{s,t}^+)^{\otimes 2}+(T_{s,t}^-)^{\otimes 2}\big]_{i,j},
\end{align}
where
\begin{align}
\bcaZ_{i,j}(\bfk_{ss})&=\sum_{u\in B(\bfk_{ss}) }\ket{u}\bra{u},\label{eq:barZijss}\\
\bcaZ_{i,j}(\bfk_{st})&=\sum_{u\in B(\bfk_{st}) }\ket{u}\bra{u}.\label{eq:barZijst}
\end{align}
Here the subscripts $i,j$ and the overbar   indicate that the operators $\bcaZ_{i,j}(\bfk_{ss})$ and $\bcaZ_{i,j}(\bfk_{st})$
act on the $n-2$ parties other than $i$ and $j$. By contrast,
the subscripts $i,j$ in $\big[(\ket{s}\bra{s})^{\otimes2}\big]_{i,j}$ and $\big[(T_{s,t}^+)^{\otimes 2}+(T_{s,t}^-)^{\otimes 2}\big]_{i,j}$  indicate that these operators act on  parties $i$ and $j$.
We perform each test with probability $1/\binom{n}{2}$, and the resulting verification operator reads
\begin{equation}\label{eq:OmegaD}
\Omega_{\bfk}=\binom{n}{2}^{-1}\sum_{i<j}P_{i,j}.
\end{equation}
The efficiency of this protocol is guaranteed by the following theorem, which is proved in Appendix~\ref{app:TheoDickeProof}. The result is summarized in Table~\ref{tab:Protocol} and illustrated in Fig.~\ref{fig:nu}.
Here it is worth pointing out that the spectral gap of $\Omega_{\bfk}$ is closely related to the spectrum of the transposition graph \cite{Chase1973,Caputo2010}, which is of interest to some researchers beyond quantum information science.

\begin{theorem}\label{thm:Dicke}
The spectral gap of $\Omega_\bfk$ reads
\begin{align}
\nu(\Omega_{\bfk})&=\begin{cases}
  1/2          & \bfk=(1,1,1),  \\
  1/3          & \bfk=(2,1),    \\
  1/(n-1) \ \  & n\ge4.
\end{cases}\label{eq:nuOmegaD}
\end{align}
To verify the Dicke state $\ket{D(\bfk)}$ within infidelity $\epsilon$ and significance level $\delta$, the number of tests required reads
\begin{align}\label{eq:NumberTestDicke}
N(\epsilon,\delta,\Omega_{\bfk})&\approx\begin{cases}
2 \epsilon^{-1} \ln\delta^{-1}          & \bfk=(1,1,1),  \\
3 \epsilon^{-1} \ln\delta^{-1}          & \bfk=(2,1),    \\
(n-1) \epsilon^{-1}\ln\delta^{-1}  \    & n\ge4.
\end{cases}
\end{align}
\end{theorem}

By construction $\Omega_{\bfk}$ is invariant under any permutation of the $n$ parties; actually  we have $\Omega_{\bfk}=P_{1,2}^S$,
where $S$ is the group of all permutations of the $n$ parties.  Therefore,   $\Omega_{\bfk}^S=\Omega_{\bfk}$ and
$\Omega_\bfk^G=\Omega_\bfk^H$, where $G=HS$, and $H$ is the group of all diagonal unitary operators of the form
$U^{\otimes n}$.
As shown in  Appendix~\ref{app:ComputeOmegaD^G}, the spectral gap of $\Omega_{\bfk}^G$ reads
\begin{align}
\nu(\Omega_{\bfk}^G)=
\frac{1}{n-1},\quad  n\ge3.\label{eq:nuOmegaD2}
\end{align}
So we have $\nu(\Omega_{\bfk}^G)=\nu(\Omega_{\bfk})$ whenever $n\geq 4$, although $\Omega_{\bfk}^G\neq \Omega_{\bfk}$ in general;
the symmetrization procedure discussed in Sec.~\ref{sec:sym} does not help in this case.

\section{Efficient verification of $W$ states}
In this section we present two more efficient protocols for verifying the $n$-qubit $W$ state  defined in \eref{eq:Wstate} \cite{Wei03,Haff05}. These protocols can reduce the number of tests quadratically with respect to the number of qubits.

\subsection{Efficient protocol based on two distinct tests}
The first protocol consists of only two distinct tests.
In the first test, called the standard test, all parties perform the Pauli-$Z$ measurements, and the test is passed if only one of the $n$ outcomes is 1.
The test projector reads
\begin{equation}\label{eq:P1}
P_1=\sum_{u\in B_n^1}|u\>\<u|,
\end{equation}
where  $B_n^1$ is  the set of strings in $\{0,1\}^n$ with Hamming weight 1.
In the other test, each of the first $n-1$ parties performs $X$ measurements; denote the outcome by 0 (1) if the measurement result is $+1$ $(-1)$.
The $n-1$ outcomes are labeled by a string $x\in \{0,1\}^{n-1}$ of $n-1$ bits, which corresponds to the product state
\begin{equation}\label{eq:alphax}
|\alpha_x\>=\frac{1}{\sqrt{2^{n-1}}}\sum_{y\in\{0,1\}^{n-1}} (-1)^{x\cdot y} |y\>\,.
\end{equation}
The reduced state of party $n$ reads
\begin{equation}
|\beta_{x}\>=\frac{|1\>+(n-1-2|x|)|0\>}{\sqrt{1+(n-1-2|x|)^2}},
\end{equation}
where $|x|$ denotes the Hamming weight of $x$.
Then party $n$ performs the two-outcome projective measurement $\big\{|\beta_{x}\>\<\beta_{x}|,I-|\beta_{x}\>\<\beta_{x}|\big\}$,
and the test is passed if the first outcome (corresponding to $|\beta_{x}\>\<\beta_{x}|$) is obtained.
The resulting test projector reads
\begin{equation}\label{eq:P2}
P_2=\sum_{x\in\{0,1\}^{n-1}} |\alpha_x\>\<\alpha_x|\otimes|\beta_{x}\>\<\beta_{x}|\,.
\end{equation}

If we perform the two tests $P_1$ and $P_2$ with probability $p$ and $1-p$, respectively, then the verification operator reads
\begin{equation}\label{eq:OmegaWn}
\Omega_{W_{n}}=pP_1+(1-p)P_2.
\end{equation}
According to \lref{lem:2TestStrategy}, the spectral gap  $\nu(\Omega_{W_{n}})$ is maximized when $p=1/2$, in which case  $\Omega_{W_n}=(P_1+P_2)/2$ and $\nu(\Omega_{W_n})=(1-\sqrt{q})/2$, where
\begin{align}\label{eq:q}
q=\|\bar{P}_1\bar{P}_2 \bar{P}_1\|=\begin{cases}
\frac{2}{5} & n=3,\\
1-h(n-3) & n\geq 4,
\end{cases}
\end{align}
with
\begin{align}\label{eq:hn}
h(n):=\frac{1}{2^{n}}\sum_{j=0}^{n}\frac{\binom{n}{j}}{1+(n-2j)^2} \,.
\end{align}
Here the second equality in \eref{eq:q} is derived in  Appendix~\ref{app:derive q}.
Therefore, we have
$\nu(\Omega_{W_{\!n}}\!)=(1/2)-(1/\sqrt{10})$ for $n=3$, and
\begin{equation}\label{eq:nuW}
\nu(\Omega_{W_n})=\frac{1-\sqrt{1-h(n-3)}}{2}> \frac{h(n-3)}{4} \quad  {\rm for}\ \, n\geq 4.
\end{equation}
The dependence of $\nu(\Omega_{W_{\!n}})$ on $n$ is  illustrated in Fig.~\ref{fig:nu}.

The function $\sqrt{n}\,h(n)$ has the following properties as proved in Appendix~\ref{app:h(n)property}.
\begin{proposition}\label{pro:sqrt{n}h(n)Monot}
$\sqrt{n}\,h(n)$ is strictly  monotonically increasing in $n$ for odd $n$ and even $n$, respectively, assuming $n\geq0$.
\end{proposition}
\begin{proposition}\label{pro:LimSqrt(n)h(n)}
When $n\rightarrow+\infty$, $\sqrt{n}\,h(n)$ converges  for odd $n$ and even $n$, respectively,
\begin{align}
&\lim_{n\rightarrow+\infty} \sqrt{2n+1}\,h(2n+1)
= \sqrt{\frac{\pi}{2}}  \tanh\Bigl(\frac{\pi}{2}\Bigr)\approx 1.15,             \label{eq:LimSqrt(n)h(n)Odd}\\
&\lim_{n\rightarrow+\infty} \sqrt{2n}\,h(2n)
                 = \sqrt{\frac{\pi}{2}}  \coth\Bigl(\frac{\pi}{2}\Bigr)\approx 1.37.                     \label{eq:LimSqrt(n)h(n)Even}
\end{align}
Here we assume that $n$ is an integer when taking the limits.
\end{proposition}
The above two propositions imply
the following inequalities:
\begin{align}
\frac{1}{2}&\leq \sqrt{n}h(n)\leq \sqrt{\frac{\pi}{2}}  \tanh\Bigl(\frac{\pi}{2}\Bigr),\quad n\geq 1\mbox{ is odd}, \label{eq:hnbound1}\\
\frac{3\sqrt{2}}{5}&\leq \sqrt{n}h(n)\leq \sqrt{\frac{\pi}{2}}  \coth\Bigl(\frac{\pi}{2}\Bigr),\quad n\geq 2\mbox{ is even}.\label{eq:hnbound2}
\end{align}
By virtue of these results, we can derive lower and upper bounds for the spectra gap, namely,
\begin{align}
\frac{1}{4\sqrt{n}}<\nu(\Omega_{W_n})<
\begin{cases}
3/(8\sqrt{n}) & n\geq 3, \ n\ne 5,\\
1/(2\sqrt{n}) & n=5;
\end{cases}
\label{eq:bound-nuW}
\end{align}
these  bounds can be improved when the parity of  $n$ is given;
see  Appendix~\ref{app:ProofEq:bound-nuW} for more details.
As a consequence of \eref{eq:bound-nuW},  the number of tests required to verify  $\ket{W_n}$ within infidelity $\epsilon$ and significance level $\delta$ satisfies
\begin{equation}
N(\epsilon,\delta,\Omega_{W_n})\leq \biggl\lceil\frac{4\sqrt{n}}{\epsilon} \ln \delta^{-1}\biggr\rceil.  \\
\end{equation}
In addition,  $\nu(\Omega_{W_n})$ admits the following limits
\begin{gather}
\!\!\lim_{n \to +\infty} \sqrt{2n+1}\nu(\Omega_{W_{2n+1}})
\!=\! \frac{\sqrt{2\pi}}{8} \coth\Bigl(\frac{\pi}{2}\Bigr)
\!\approx\! 0.342,\label{eq:Lim sqrt(2n+1)nu}\\
\lim_{n \to +\infty} \sqrt{2n}\nu(\Omega_{W_{2n}})
= \frac{\sqrt{2\pi}}{8} \tanh\Bigl(\frac{\pi}{2}\Bigr)
\approx 0.287,\label{eq:Lim sqrt(2n)nu}
\end{gather}
as proved in Appendix~\ref{app:Proof Lim sqrt(n)nu}.
When $n\gg1$, we have
\begin{align}
&\nu(\Omega_{W_n})\approx\begin{cases}
\frac{\sqrt{2\pi}}{8\sqrt{n} } \coth\bigl(\frac{\pi}{2}\bigr)\approx \frac{0.342}{\sqrt{n}}   & n \ $is odd$,\\[0.8ex]
\frac{\sqrt{2\pi}}{8\sqrt{n} } \tanh\bigl(\frac{\pi}{2}\bigr)\approx \frac{0.287}{\sqrt{n} }\ \ & n \ $is even$;
\end{cases}\label{eq:lim-nuW}\\
&N(\epsilon,\delta,\Omega_{W_n})\approx\begin{cases}
2.93\sqrt{n} \epsilon^{-1}\ln \delta^{-1}    & n \ $is odd$,\\
3.48\sqrt{n} \epsilon^{-1}\ln \delta^{-1}  \ & n \ $is even$.
\end{cases}
\end{align}
These results are summarized in Table~\ref{tab:Protocol} and illustrated in Fig.~\ref{fig:nu}.
Compared with the protocol in Ref.~\cite{Liu19} which achieves $\nu=1/(n-1)$ with $\binom{n}{2}$ distinct tests when $n\geq 4$ (cf.~Sec.~\ref{sec:DickeVerify}), the current protocol achieves a much better scaling behavior in $n$ and a higher efficiency whenever $n\geq15$, although only two distinct tests are required.

\subsection{\label{sec:symW}
Higher efficiency from symmetrization}
The efficiency of the above protocol can be improved by applying the symmetrization procedure described in Sec.~\ref{sec:sym}. Let $G$  be the group generated by all permutations of the $n$ qubits and diagonal unitary operators of the form $U^{\otimes n}$. Consider the symmetrized verification operator
\begin{equation}
\Omega_{W_{n}}^G=pP_1^G+(1-p)P_2^G=pP_1+(1-p)P_2^G,
\end{equation}
Note that $P_1^G=P_1$ is a projector, but $P_2^G$ is not a projector. So
\lref{lem:2TestStrategy} is not applicable, and  here the optimal choice of $p$  is not $1/2$ in contrast to \eref{eq:OmegaWn}. Denote by $\caH_1$ the support of $P_1$ and by $\caH_2$ the orthogonal complement of $\caH_1$. Then $\caH_1$ and $\caH_2$ are invariant subspaces of $G$. In addition, $G$ has two inequivalent irreducible components in $\caH_1$:  one component is spanned by $|W_n\>$ and is one dimensional; the other component consists of all vectors in $\caH_1$ that are orthogonal to $|W_n\>$. Each irreducible component in $\caH_2$ is not equivalent to any irreducible component in $\caH_1$.
Consequently, $P_2^G$ is block diagonal with respect to $\caH_1$ and $\caH_2$; in addition, $P_1 \bar{P}_2^G P_1$ is proportional to a projector.  Let $R$ be the subgroup of $G$ generated by $\diag(1,\rme^{2\pi\rmi/(n+1)})^{\otimes n}$ and a cyclic permutation of order $n$; note that $R$ has order $n(n+1)$. By virtue of  Proposition~\ref{pro:GapSym3}, it is not difficult to verify that $\Omega_{W_{n}}^R=\Omega_{W_{n}}^G$, given that $P_1$ is invariant under all permutations, while $P_2$ is invariant under permutations of the first $n-1$ parties. Therefore, the strategy $\Omega_{W_{n}}^G$ can be realized using $n^2+n+1$ distinct projective tests.

As derived in Appendix~\ref{app:Wsymmetrization}, we have
\begin{align}\label{eq:TraceP1P2}
\tr(P_1 P_2^G)=\tr(P_1 P_2)=n-1-(n-2)h(n-1),
\end{align}
where $h(n)$ is defined in \eref{eq:hn}.
It follows that
\begin{equation}
\|P_1 \bar{P}_2^G P_1\|=\frac{(n-2)[1-h(n-1)]}{n-1}\leq 1-h(n-1).
\end{equation}
Let
\begin{equation}\label{eq:OptProbW}
p=\frac{1-\|P_1 \bar{P}_2^G P_1\|}{2-\|P_1 \bar{P}_2^G P_1\|}=\frac{1+(n-2)h(n-1)}{n+(n-2)h(n-1)};
\end{equation}
then we have
\begin{align}
\lambda_2(\Omega_{W_{n}}^G)&=1-p=\frac{n-1}{n+(n-2)h(n-1)},  \label{eq:lambda(Omega^G)}\\
\nu(\Omega_{W_{n}}^G)&=p=\frac{1+(n-2)h(n-1)}{n+(n-2)h(n-1)}>\frac{1}{\sqrt{n}+1},\label{eq:nu(Omega^G)Bound}
\end{align}
as shown in  Appendix~\ref{app:Wsymmetrization}. In addition, by virtue of Proposition~\ref{pro:LimSqrt(n)h(n)} as well as  Eqs.~(\ref{eq:hnbound1}) and \eqref{eq:hnbound2}, we can deduce the following limits,
\begin{gather}
\!\!\lim_{n \to +\infty} \sqrt{2n+1}\,\nu \big(\Omega_{W_{2n+1}\!}^G\big)
=  \sqrt{\frac{\pi}{2}} \coth\Bigl(\frac{\pi}{2}\Bigr)
\approx 1.37,\\
\lim_{n \to +\infty} \sqrt{2n}\,  \nu \bigl(\Omega_{W_{2n}}^G\bigr)
= \sqrt{\frac{\pi}{2}} \tanh\Bigl(\frac{\pi}{2}\Bigr)
\approx 1.15.
\end{gather}
Numerical calculation shows that
a good approximation of $\nu(\Omega_{W_n})$ can be expressed as follows,
\begin{align}
&\nu(\Omega_{W_n})\approx\begin{cases}
 \frac{1.37}{\sqrt{n}+1.37}   & n \ $is odd$,\\
 \frac{1.15}{\sqrt{n}+1.11 }\ \ & n \ $is even$.
\end{cases}
\end{align}
When $n\gg 1$, we have $\nu(\Omega_{W_{n}}^G)\approx 4\nu(\Omega_{W_n})$, so the symmetrization
procedure can improve the efficiency  by about four times.

A comparison of the strategies $\Omega_{\bfk}$, $\Omega_{W_{n}}$, and $\Omega_{W_{n}}^G$ indicates that  $\Omega_{W_{n}}^G$ has the largest spectral gap and thus the highest
efficiency for all $n\geq3$ [cf.~\eqssref{eq:nuOmegaD}{eq:bound-nuW}{eq:nu(Omega^G)Bound}], as illustrated in Fig.~\ref{fig:nu}.
The strategy $\Omega_{W_{n}}$ requires only two distinct tests, which is much fewer than the number $O(n^2)$ of distinct tests
required by the other two strategies. On the other hand, the strategies $\Omega_{W_{n}}$ and $\Omega_{W_{n}}^G$ only apply to
$W$ states, while the strategy $\Omega_{\bfk}$ applies to all  qudit (including qubit) Dicke states.

\section{Nearly Optimal Verification of the three-qubit $W$ state}
In this section we construct a nearly optimal protocol for verifying  the three-qubit $W$ state $|W_3\>$ \cite{Dur00} shared by Alice, Bob, and Charlie. Before presenting this protocol, it is instructive to set an upper bound for the spectral gap of any verification operator based on LOCC.

According to Ref.~\cite{Wang19}, for a normalized two-qubit entangled pure state $s_0|00\>+s_1|11\>$
with Schmidt coefficients $s_0, s_1$ ($0<s_0,s_1<1$ and $s_0^2+s_1^2=1$),  the maximum spectral gap of any verification
operator based on LOCC or separable measurements is $1/(1+s_0s_1)$. With respect to the partition between Alice and the other two parties, $|W_3\>$  can be regarded as a two-qubit state in a proper subspace and has two Schmidt coefficients equal to $\sqrt{1/3}$ and $\sqrt{2/3}$, respectively. Therefore, the  spectral gap of any verification
operator based on LOCC or separable measurements is upper bounded by
\begin{equation}
\frac{1}{1+\sqrt{2/9}}=\frac{9-3\sqrt{2}}{7}\approx 0.6796.
\end{equation}
If each test  of the verification strategy can be realized by LOCC with one-way communication, then the upper bound can be reduced to $2/3$ according to Refs.~\cite{Wang19,Yu19}.

\subsection{Nearly optimal verification protocol}
To start with, we construct an efficient protocol using three distinct tests.
In the first test, all three parties perform $Z$ measurements, and
the test is passed if only one of the three outcomes is $1$.
The test projector reads
\begin{equation}\label{eq:W3P1}
P_1 = |001\>\<001|+ |010\>\<010|+ |100\>\<100|,
\end{equation}
which is a special case of the projector defined in \eref{eq:P1}.
The other two tests are based on adaptive local projective measurements.
The second test  $P_2$ is defined in \eref{eq:P2} with $n=3$ and  has the form
\begin{align}
P_2=& \; X^+X^+\otimes|\gamma_+\>\<\gamma_+| + X^-X^-\otimes|\gamma_-\>\<\gamma_-| \nonumber\\
& \, + (X^+X^- + X^-X^+) \otimes |1\>\<1| ,
\end{align}
where $|\gamma_\pm\>=\frac{1}{\sqrt{5}}(2|0\>\pm|1\>)$,
$X^\pm=|\pm\>\<\pm|$, and $|\pm\>=\frac{1}{\sqrt{2}}(|0\>\pm|1\>)$ are eigenstates of the operator $X$.
For the third test, Alice performs $Z$ measurement and send her outcome to Bob and Charlie.
If the outcome of Alice is $1$, so that the normalized reduced state of Bob and Charlie is $\ket{00}$ (if $|W_3\>$ is measured), then
both Bob and Charlie perform $Z$ measurement, and the test is passed if their outcomes are both $0$.
If the outcome of Alice is $0$, so that the normalized reduced state reads $\frac{1}{\sqrt{2}}(\ket{01}+\ket{10})$ (if $|W_3\>$ is measured),
then  both Bob and Charlie perform $X$ measurement, and the test is passed if their outcomes coincide.
The resulting test projector reads
\begin{equation}
P_3 = |100\>\<100| + |0\>\<0| \otimes (X^+X^+ + X^-X^-) .
\end{equation}
Note that the three test projectors $P_1$, $P_2$, and $P_3$ have ranks 3, 4, and 3, respectively.

If we perform the three tests $P_1$, $P_2$, and $P_3$ with
probabilities $p_1$, $p_2$, and $1-p_1-p_2$, respectively, then the verification operator is given by
\begin{equation}
\Omega_{\1}=p_1P_1+p_2P_2+(1-p_1-p_2)P_3.
\end{equation}
Note that this strategy can be realized using local projective measurements with one-way communication.
Numerical calculation shows that $\lambda_2(\Omega_{\1})\geq 0.695$,
and the lower bound is approximately saturated when $p_1\approx0.246$ and $p_2\approx0.444$, in which
case we have $\nu(\Omega_{\1})\approx 0.305$.

The efficiency of the above protocol can be improved by applying the symmetrization procedure described in Sec.~\ref{sec:sym}.
Let $G$ be the group generated by the six permutations and diagonal unitary operators of the form $U^{\otimes3}$. Then $G$ has six irreducible components, all of which are inequivalent.
Let
\begin{equation}
\begin{aligned}
&|\tau_0\>:=|000\>,  \qquad  |\tau_1\>:=|111\>, \\
&|\tau_2\> :=(|001\>-|010\>)/\sqrt{2},          \\
&|\tau_3\> :=(|001\>+|010\>-2|100\>)/\sqrt{6},  \\
&|\tau_4\> :=(|011\>+|101\>+|110\>)/\sqrt{3},   \\
&|\tau_5\>:=(|011\>-|101\>)/\sqrt{2},           \\
&|\tau_6\>:=(|011\>+|101\>-2|110\>)/\sqrt{6}.
\end{aligned}
\end{equation}
Then  four one-dimensional irreducible components of $G$ are spanned by $|W_3\>$, $|\tau_0\>$, $|\tau_1\>$, $|\tau_4\>$, respectively.
One two-dimensional component is spanned by $|\tau_2\>$ and $|\tau_3\>$, and the other   two-dimensional component is spanned by $|\tau_5\>$ and $|\tau_6\>$. Given any verification operator $\Omega$ for $|W_3\>$, then $\Omega^G$ has the form
\begin{align}\label{eq:SymmOmegaPPT}
\Omega^G
= &\,|W_3\>\<W_3| +\mu_0|\tau_0\>\<\tau_0| +\mu_1|\tau_1\>\<\tau_1| +\mu_4|\tau_4\>\<\tau_4| \nonumber\\
&  +\mu_2\left(|\tau_2\>\<\tau_2| +|\tau_3\>\<\tau_3|\right)+\mu_3\left(|\tau_5\>\<\tau_5| +|\tau_6\>\<\tau_6|\right)
\end{align}
according to \eref{eq:OmegaSymProj},
where $0\leq \mu_0,\mu_1,\mu_2,\mu_3,\mu_4\leq1$. On the other hand, any verification operator of this form is $G$-invariant.

Let $K$ be the subgroup of $G$ that is generated by six permutations and $U_{\pi/2}^{\otimes 3}$ with $U_{\pi/2}=\diag(1,\rmi)$; note that $K$ has order 24.  Then $K$ has the same number of irreducible components as $G$, so $\Omega^K=\Omega^G$ for any verification operator of $|W_3\>$ according to Proposition~\ref{pro:GapSym3}. In addition, if $\Omega$ can be realized by $m$ distinct projective  tests, then $\Omega^K$ can be realized by at most $24m$ distinct projective tests.

Consider the verification operator
\begin{equation}
\Omega_{\2}:=\Omega_{\1}^K=p_1P_1+p_2P_2^K+(1-p_1-p_2)P_3^K;
\end{equation}
note that $P_1^K=P_1$.
Each test operator $P_j^K$ for $j=1,2,3$  has the form in \eref{eq:SymmOmegaPPT}
with at most five distinct eigenvalues. The parameter vectors $\mu=(\mu_0,\mu_1,\mu_2,\mu_3,\mu_4)$ associated with the three test operators are respectively given by
\begin{equation}	
\begin{aligned}
\mu&=(0,0,1,0,0) \quad  {\rm for}\ \, P_1^K, \\
\mu&=\frac{1}{15}(6,9,3,8,8) \quad  {\rm for}\ \, P_2^K,\\
\mu&=\frac{1}{6}(3,0,3,1,1) \quad  {\rm for}\ \, P_3^K.
\end{aligned}
\end{equation}
Therefore, the second largest eigenvalue of $\Omega_{\2}$ reads
\begin{align}
\lambda_2(\Omega_{\2})
=&\max_{\substack{p_1\!,\,p_2\geq0\\p_1+p_2\leq1}}
\bigg\{ \frac{5-5p_1-p_2}{10}, \frac{5-5p_1+11p_2}{30},\nonumber\\
&\  \frac{5+5p_1-3p_2}{10}, \frac{3}{5}p_2 \bigg\}
\geq\frac{3}{8} \,.
\end{align}
The bound is saturated iff $p_1=1/8$ and $p_2=5/8$, in which case we have
\begin{equation}\label{eq:OurStraforW3}
\Omega_{\2}=|W_3\>\<W_3|+\frac{3}{8}\big(\openone-|W_3\>\<W_3|\big),
\end{equation}
and
\begin{equation}\label{eq:nu5/8}
\nu(\Omega_{\2}) =\frac{5}{8}\,, \qquad
N(\epsilon,\delta,\Omega_{\2})\approx\frac{8}{5\epsilon} \ln \delta^{-1}.
\end{equation}
Compared with the protocol in Ref.~\cite{Liu19} which achieves $\nu=1/3$ (cf.~Sec.~\ref{sec:sym}), this  protocol has a much higher efficiency.
In addition, the spectral gap is only 8.04\% smaller than the upper bound $\nu(\Omega)\leq(9-3\sqrt{2})/7$ for strategies based on LOCC or separable measurements. Accordingly, the number of tests required by the strategy $\Omega_{\2}$ is only 8.74\% more than the optimal strategy based on separable measurements.

\begin{figure*}
\begin{center}
	\includegraphics[width=15.95cm]{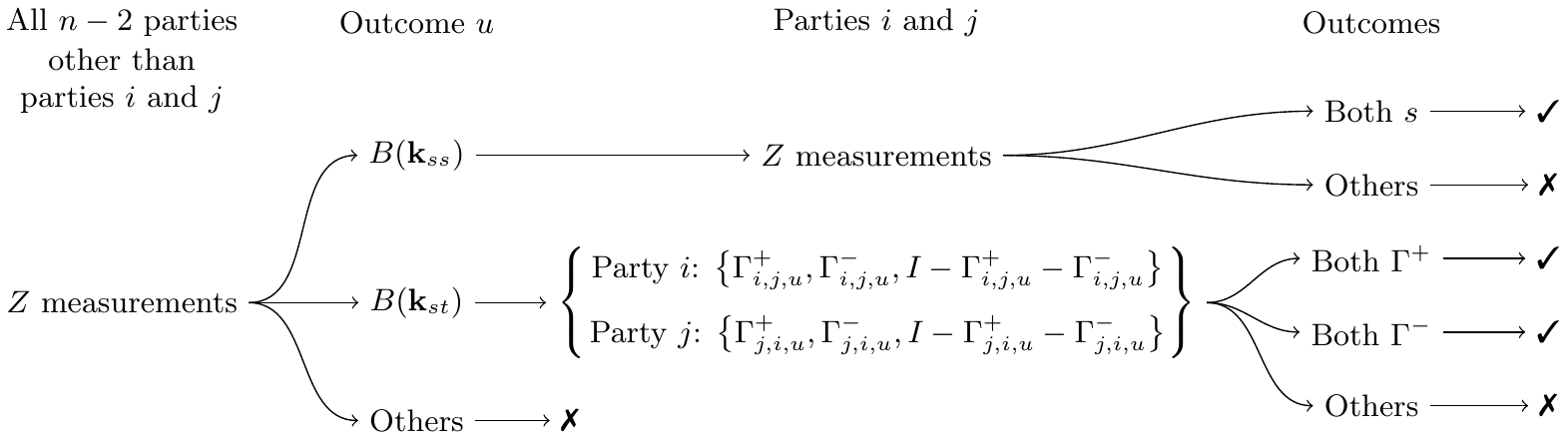}
	\caption{\label{fig:measurement-III} Schematic view of the test $P^{\phi}_{i,j}$ used to verify the phased Dicke state $\ket{D_{\phi}(\bfk)}$.
All $n-2$ parties other than parties $i$ and $j$ first perform the generalized Pauli-$Z$ measurement and send the outcome $u$ to parties $i$ and $j$.
Conditioned on this outcome, parties $i$ and $j$ then perform  suitable projective measurements. The outcomes corresponding to passing the test are marked by "\cmark".
}
\end{center}
\end{figure*}

\subsection{Additional applications}
The strategy  $\Omega_{\2}$ is homogeneous and so  can be applied to fidelity estimation \cite{ZhuEVQPSlong19}.
Note that the passing probability of  any state $\rho$ is related to its fidelity with the target state $|W_3\>$ as follows, $\tr(\rho\Omega_{\2})=\frac{5}{8}\<W_3|\rho|W_3\>+\frac{3}{8}$, which implies that
\begin{equation}
F=\<W_3|\rho|W_3\>=\frac{8}{5}\tr(\rho\Omega_{\2})-\frac{3}{5}.
\end{equation}
According to Ref.~\cite{ZhuEVQPSlong19}, the standard deviation of this
estimation is given by $\Delta F\!=\!\sqrt{(1-F)(F+3/5)/N}$, where $N$ is the number of tests performed.

Besides fidelity estimation, our protocol in \eref{eq:OurStraforW3} is also useful for state verification in the adversarial scenario, in which case the state to be verified  is prepared by a potentially malicious adversary \cite{ZhuEVQPSshort19,ZhuEVQPSlong19}.
If there is no restriction on the accessible measurements, the optimal strategy for verifying $|\Psi\>$ in the adversarial
scenario can be chosen to be homogeneous,
\begin{equation}\label{eq:HomoStra}
\Omega=|\Psi\>\<\Psi|+\lambda_2(\Omega)\big(\openone-|\Psi\>\<\Psi|\big).
\end{equation}
According to Refs.~\cite{ZhuEVQPSshort19,ZhuEVQPSlong19},
in the high-precision limit  $\epsilon,\delta\rightarrow 0$,  the minimal number of tests required to verify $|\Psi\>$ reads    (assuming $\lambda_2(\Omega)>0$),
\begin{equation}\label{eq:NumTestAdv}
N\approx[\lambda_2(\Omega)\epsilon\ln \lambda_2(\Omega)^{-1}]^{-1}  \ln \delta^{-1}.
\end{equation}
This number is minimized when $\lambda_2(\Omega)=1/\rme$, which yields $N\approx \rme\epsilon^{-1} \ln \delta^{-1}$.
Since our verification strategy $\Omega_{\2}$ for $|W_3\>$ is homogeneous with $\lambda_2(\Omega_{\2})=3/8$,
it can be applied to the adversarial scenario directly.
For high-precision state verification, the number of tests required reads
$N\approx 2.7188 \epsilon^{-1} \ln\delta^{-1}$,
which is only about 0.02\% more than the optimal strategy.
When  $\epsilon,\delta$ are   small but not infinitesimal (say $\epsilon,\delta \leq 0.01$), our strategy is still nearly optimal.

\section{\label{sec:PhasedDicke}Verification of phased Dicke states}
In this section we consider the verification of phased Dicke states \cite{Krammer09,Chiuri10}, which have  the form
\begin{equation}\label{eq:quditDstateGen}
\ket{D_{\phi}(\bfk)}=\frac{1}{\sqrt{m}} \sum_{u\in B(\bfk)}\rme^{\rmi\phi(u)}|u\>,
\end{equation}
where $m=|B(\bfk)|=n!/\big(\prod _{j=0}^r k_j!\big)$ and the phase $\phi(u)$ is a real-valued function of the sequence $u$.

Similar to the verification protocol for Dicke states, our protocol for $|D_{\phi}(\bfk)\>$ consists of $\binom{n}{2}$ distinct tests based on adaptive local projective measurements.
Each test is associated with a pair of parties among the $n$ parties.
The test $P^{\phi}_{i,j}$ associated with parties $i$ and $j$  is realized  as follows as illustrated in Fig.~\ref{fig:measurement-III}.
All $n-2$ parties other than parties $i$ and $j$ perform the generalized Pauli-$Z$ measurements, and their outcomes are labeled by a sequence $u$ of $n-2$ symbols, which corresponds to the product state $|u\>$. The measurements of parties $i$ and $j$ depend on the outcome $u$, and  we need to distinguish three cases.
Recall that $\bfk_{st}$ and $\bfk_{ss}$ are defined in  Eqs.~(\ref{eq:kst2}) and~(\ref{eq:kss}), respectively.
Suppose $k_0,k_1,\dots,k_g\geq2$ and $k_{g+1}=k_{g+2}=\cdots k_r=1$, where $-1\leq g\leq r$.
\begin{enumerate}
\item[1.] $u\in B(\bfk_{ss})$ with $0\leq s\leq g$. \\
In this case, the normalized reduced state of parties $i$ and $j$ reads $|s\>_i|s\>_j$ up to an irrelevant phase factor (if the target phased Dicke state is measured).
Then the two parties both perform $Z$ measurement, and the test is passed if they both obtain outcome~$s$.

\item[2.] $u\in B(\bfk_{st})$ with $0\leq s<t\leq r$.\\
In this case, the normalized reduced state of parties $i$ and $j$ reads,
\begin{align}
&\frac{1}{\sqrt{2}}\bigl[|s\>_i|t\>_j+\rme^{\rmi\theta(i,j,u)}|t\>_i|s\>_j\bigr], \\
&\theta(i,j,u):=\phi(v(j,i,u)) -\phi(v(i,j,u)).
\end{align}
Here  $v(i,j,u), v(j,i,u)\in B(\bfk)$ are defined as follows,
\begin{align}
v_i(i,j,u)&=s,  \quad v_j(i,j,u)=t, \quad  v_{\overline{{i,j}}}(i,j,u)=u,\label{eq:viju}\\
v_i(j,i,u)&=t,  \quad v_j(j,i,u)=s, \quad  v_{\overline{{i,j}}}(j,i,u)=u, \label{eq:vjiu}
\end{align}
where $v_{\overline{{i,j}}}(i,j,u)$ means the subsequence of $v(i,j,u)$ without the $i$th and $j$th components, and $v_{\overline{{i,j}}}(j,i,u)$ is defined in the same way. Note that the parameters  $s$ and $t$ are determined by $u$.
Then parties $i$ and $j$ perform projective measurements
$\big\{\Gamma_{i,j,u}^+,\Gamma_{i,j,u}^-,I-\Gamma_{i,j,u}^+ - \Gamma_{i,j,u}^-\big\}$
and $\big\{\Gamma_{j,i,u}^+,\Gamma_{j,i,u}^-,I-\Gamma_{j,i,u}^+ - \Gamma_{j,i,u}^-\big\}$, respectively, where
\begin{align}
\Gamma_{i,j,u}^+=\frac{1}{2}\big[\ket{s}+\rme^{\rmi\theta(i,j,u)/2}\ket{t}\big]\!
\big[\bra{s}+\rme^{-\rmi\theta(i,j,u)/2}\bra{t}\big], \\
\Gamma_{i,j,u}^-=\frac{1}{2}\big[\ket{s}-\rme^{\rmi\theta(i,j,u)/2}\ket{t}\big]\!
\big[\bra{s}-\rme^{-\rmi\theta(i,j,u)/2}\bra{t}\big],
\end{align}
and $\Gamma_{j,i,u}^\pm$ are defined in a similar way with $\theta(i,j,u)$ replaced by  $\theta(j,i,u)=-\theta(i,j,u)$.
The test is passed if they both obtain the first outcome (corresponding to $\Gamma^+$) or if they both obtain the second outcome (corresponding to $\Gamma^-$).

\item[3.] Other cases.\\
The state cannot be the target state $\ket{D_{\phi}(\bfk)}$, so the test is not passed.
\end{enumerate}
The resulting test projector reads
\begin{align}
P^{\phi}_{i,j}&=  \sum_{s=0}^g \bcaZ_{i,j}(\bfk_{ss}) \otimes \big[(\ket{s}\bra{s})^{\otimes2}\big]_{i,j}+\sum_{s<t}\sum_{u\in B(\bfk_{st})}|u\>\<u| \nonumber\\
&\quad  \otimes
\bigl(\Gamma_{i,j,u}^+\otimes \Gamma_{j,i,u}^+ + \Gamma_{i,j,u}^- \otimes \Gamma_{j,i,u}^-\bigr)_{i,j},
\end{align}
where $\bcaZ_{i,j}(\bfk_{ss})$ is the projector defined in \eref{eq:barZijss}.
Each test is performed with probability $1/\binom{n}{2}$, and the resulting verification operator reads
\begin{equation}\label{eq:OmegaD'}
\Omega^{\phi}_{\bfk}=\binom{n}{2}^{-1}\sum_{i<j}P^{\phi}_{i,j}.
\end{equation}
The efficiency of this protocol is guaranteed by the following theorem, which is proved in Appendix~\ref{app:TheoDstateGenProof}. As in the case of Dicke states,  the spectral gap of $\Omega^{\phi}_{\bfk}$ is closely related to the spectrum of the transposition graph \cite{Chase1973,Caputo2010}.
\begin{theorem}\label{thm:DstateGen}
The spectral gap of $\Omega^{\phi}_{\bfk}$ is the same as that of $\Omega_{\bfk}$ in \eref{eq:nuOmegaD}, namely,
\begin{equation}\label{eq:nuOmegaD'}
 \nu\big(\Omega^{\phi}_{\bfk}\big)=\nu(\Omega_{\bfk})=\begin{cases}
 1/2          & \bfk=(1,1,1),  \\
 1/3          & \bfk=(2,1),    \\
 1/(n-1) \ \  & n\ge4.
 \end{cases}
\end{equation}
To verify the phased Dicke state $\ket{D_{\phi}(\bfk)}$ within infidelity $\epsilon$ and significance level $\delta$, the number of tests required reads
\begin{align}\label{eq:NOmegaD'}
N\big(\epsilon,\delta,\Omega^{\phi}_{\bfk}\big)&\approx\begin{cases}
2 \epsilon^{-1} \ln\delta^{-1}          & \bfk=(1,1,1),  \\
3 \epsilon^{-1} \ln\delta^{-1}          & \bfk=(2,1),    \\
(n-1) \epsilon^{-1}\ln\delta^{-1}  \    & n\ge4.
\end{cases}
\end{align}
\end{theorem}

\begin{figure*}
\begin{center}
	\includegraphics[width=15.9cm]{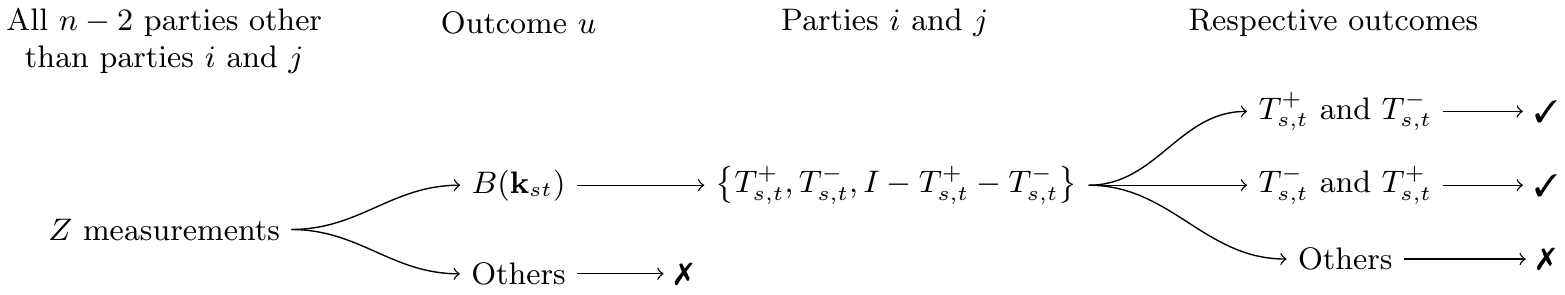}
	\caption{\label{fig:measurement-II} Schematic view of the test $P_{i,j}^{\AS}$ used to verify the $n$-partite antisymmetric basis state $|\AS_n\>$.
All $n-2$ parties other than parties $i$ and $j$ first perform the generalized Pauli-$Z$ measurement and send the outcome $u$ to parties $i$ and $j$.
Conditioned on this outcome, parties $i$ and $j$ then perform  suitable projective measurements. The outcomes corresponding to passing the test are marked by "\cmark".
}
\end{center}
\end{figure*}

\section{Optimal verification of antisymmetric basis states}
Finally,  we consider the verification of the $n$-partite antisymmetric basis state, also known as the Slater determinant
state \cite{Denni01,Zanar02,Bravyi03}. It has the following form
\begin{equation}\label{eq:AntisymState}
|\AS_n\>=\frac{1}{\sqrt{n!}}\sum_{j_1,j_2,\dots,j_n} \tilde\epsilon_{j_1,\dots,j_n} \ket{j_1-1}\otimes\cdots\otimes\ket{j_n-1} ,
\end{equation}
where $j_1,j_2\dots,j_n\in\{1,2,\dots,n\}$ and $\tilde\epsilon_{j_1,j_2,\dots,j_n}$ is the Levi-Civita symbol.
Note that $|\AS_n\>$ can be regarded as a bipartite maximally entangled state of Schmidt rank $n$ between one party and the other parties. So the spectral gap of any verification operator based on LOCC or separable measurements is upper bounded by $n/(n+1)$ according to  known results on the verification of  maximally entangled states \cite{HayaMT06,Haya09,ZhuH19O}. Here we shall show that this upper bound can be saturated for any antisymmetric basis state with $n\geq 2$. When $n=2$, the state $|\AS_2\>$ is a singlet and can be verified using protocols for bipartite pure states proposed in Refs.~\cite{HayaMT06,Haya09,ZhuH19O,LHZ19,Wang19,Yu19}. Here we focus on the multipartite case with $n\geq3$.

\subsection{Efficient verification protocol}
Note that the antisymmetric basis state $|\AS_n\>$ in \eref{eq:AntisymState} is a  special case of phased Dicke states in \eref{eq:quditDstateGen} with $\bfk=(1,1,\ldots, 1)$ and $\phi(u)=1$ ($-1$) if $u$ is an even (odd) permutation of $0,1,2,\ldots, n-1$. Therefore, $|\AS_n\>$
can be verified using the strategy presented in \eref{eq:OmegaD'} tailored to this specific case.  Here we shall construct a variant  protocol that also consists of $\binom{n}{2}$ distinct tests, and each test is associated with a pair of parties.
The test $P_{i,j}^{\AS}$ associated with parties $i$ and $j$ is illustrated in Fig.~\ref{fig:measurement-II} and realized as follows. All $n-2$ parties other than parties $i$ and $j$ perform the generalized Pauli-$Z$ measurements, and their outcomes are labeled by a sequence $u$ of $n-2$ symbols, which corresponds to the product state $|u\>$. The measurements of parties $i$ and $j$ depend on the outcome $u$, and  we need to distinguish two cases.
\begin{enumerate}
	\item[1.] $u\in B(\bfk_{st})$ with $0\leq s<t\leq n-1$.\\
	In this case, the normalized reduced state of parties $i$ and $j$ reads $\frac{1}{\sqrt{2}}(|s\>_i|t\>_j-|t\>_i|s\>_j)$.
	Then the two parties both perform the projective measurement $\big\{T_{s,t}^+,T_{s,t}^-,I-T_{s,t}^+ -T_{s,t}^-\big\}$, where
	$T_{s,t}^+$ and $T_{s,t}^-$ are projectors defined in \eqsref{eq:Pst+}{eq:Pst-}.
	The test is passed if one of them obtains the first outcome (corresponding to $T_{s,t}^+$) and the other one obtains the second outcome (corresponding to $T_{s,t}^-$).
	
	\item[2.] Other cases.\\
	The state cannot be the target state $\ket{\AS_n}$, so the test is not passed.
\end{enumerate}
The resulting test projector reads
\begin{equation}\label{eq:PijAS}
P_{i,j}^{\AS}=\sum_{s<t}  \bcaZ_{i,j}(\bfk_{st}) \otimes \left(T_{s,t}^+\otimes T_{s,t}^- + T_{s,t}^-\otimes T_{s,t}^+\right)_{i,j},
\end{equation}
where $\bcaZ_{i,j}(\bfk_{st})$ is defined in \eref{eq:barZijst} and acts on the tensor product space of all parties other than $i$ and $j$.

We perform each test with probability $1/\binom{n}{2}$, and the resulting verification operator reads
\begin{equation}\label{eq:OmegaAS}
\Omega_{\AS_n}=\binom{n}{2}^{-1} \sum_{i<j}P_{i,j}^{\AS}.
\end{equation}
The efficiency of this protocol is guaranteed by the following theorem, which is proved in Appendix~\ref{app:TheoAntiStateProof}.
\begin{theorem}\label{thm:Antisymmetric}
	The spectral gap of $\Omega_{\AS_n}$ with $n\geq3$ reads
	\begin{equation}\label{eq:nuOmegaAS}
	\nu\big(\Omega_{\AS_n}\big) =  \frac{1}{n-1}.
	\end{equation}
	To verify the antisymmetric basis state $\ket{\AS_n}$ within infidelity $\epsilon$ and significance level $\delta$, the number of tests required reads
	\begin{equation}\label{eq:NumberTestAS}
	N\big(\epsilon,\delta,\Omega_{\AS_n}\big) \approx \frac{n-1}{\epsilon}\ln\delta^{-1}.
	\end{equation}
\end{theorem}

Incidentally, the measurement $\big\{T_{s,t}^+,T_{s,t}^-,I-T_{s,t}^+ -T_{s,t}^-\big\}$ employed in the above verification protocol can be replaced by the alternative $\big\{\tilde{T}_{s,t}^+,\tilde{T}_{s,t}^-,I-\tilde{T}_{s,t}^+ -\tilde{T}_{s,t}^-\big\}$, where
\begin{align}
\tilde{T}_{s,t}^+&=\frac{1}{2}(\ket{s}+\rmi\ket{t})(\bra{s}-\rmi\bra{t}),\\
\tilde{T}_{s,t}^-&=\frac{1}{2}(\ket{s}-\rmi\ket{t})(\bra{s}+\rmi\bra{t}).
\end{align}
Accordingly, the test projector $P_{i,j}^{\AS}$ is replaced by
\begin{equation}
\tilde{P}_{i,j}^{\AS}=\sum_{s<t}  \bcaZ_{i,j}(\bfk_{st}) \otimes \bigl(\tilde{T}_{s,t}^+\otimes \tilde{T}_{s,t}^- + \tilde{T}_{s,t}^-\otimes \tilde{T}_{s,t}^+\bigr)_{i,j},
\end{equation}
and the resulting verification operator $\tilde{\Omega}_{\AS_n}$ is given by \eref{eq:OmegaAS} with $P_{i,j}^{\AS}$  replaced by $\tilde{P}_{i,j}^{\AS}$.
This verification strategy
is a special case of the strategy presented in Sec.~\ref{sec:PhasedDicke} (tailored to the antisymmetric basis state). According to Theorems~\ref{thm:DstateGen} and \ref{thm:Antisymmetric}, we have
\begin{equation}
\nu\big(\tilde{\Omega}_{\AS_n}\big) =  \frac{1}{n-1}=\nu\big(\Omega_{\AS_n}\big).
\end{equation}
Therefore, the two strategies $\Omega_{\AS_n}$ and $\tilde{\Omega}_{\AS_n}$ are equally efficient.

\subsection{\label{sec:OptimalASn}Optimal verification protocol based on symmetrization}
Let $\tilde{H}$ be the group of all unitary transformations of the form $U^{\otimes n}$ with $U\in \rmU(\bbC^n)$ (here $U$ is not required to be diagonal), let $S$ be the group of all permutations of the $n$ parties, and let $\tilde{G}=\tilde{H}S$. Then the  projector onto the antisymmetric basis state $|\AS_n\>$ is invariant under $\tilde{G}$.  Therefore, we can construct a symmetrized  strategy $\Omega_{\AS_n}^{\tilde{G}}$ according to Sec.~\ref{sec:sym}. Similar to $\Omega_{\bfk}$, by construction $\Omega_{\AS_n}$ is invariant under $S$,
so we have $\Omega_{\AS_n}^S=\Omega_{\AS_n}$ and $\Omega_{\AS_n}^{\tilde{G}}=\Omega_{\AS_n}^{\tilde{H}}$.
For the convenience of practical applications,  the group $\rmU(\bbC^n)$ used to construct $\tilde{H}$ can also be replaced by a unitary $t$-design with $t=n$ \cite{DankCEL09,Gross07}.

To determine $\Omega_{\AS_n}^{\tilde{G}}$, note that
$\tilde{H}$ is a representation of $\rmU(\bbC^n)$ and $S$ is a representation of the symmetric group $\scrS_n$ of $n$ letters. Accordingly, $\tilde{G}$ is a representation of $\rmU(\bbC^n)\times \scrS_n$. By Schur-Weyl duality \cite{Weyl1931,Proc2007}, all the  irreducible components of $\tilde{G}$ in $(\bbC^n)^{\otimes n}$ are multiplicity free, and each irreducible component is labeled by a partition of $n$. Meanwhile,
$(\bbC^n)^{\otimes n}$ has the following decomposition
\begin{equation}
(\bbC^n)^{\otimes n}=\bigoplus_{\mu \vdash n} \caH_\mu=\bigoplus_{\mu \vdash n} \caW_\mu\otimes \caS_\mu,
\end{equation}
where the notation  $\mu \vdash n$ means $\mu=(\mu_1, \mu_2,\ldots, \mu_n)$ is  a partition of $n$, which means $\mu_j$ are nonnegative integers arranged in decreasing order and sum up to $n$. Here $\caW_\mu$ carries the irreducible representation of the unitary group  $\rmU(\bbC^n)$, while $\caS_\mu$ carries the irreducible representation of the symmetric  group $\scrS_n$.
Let  $D_\mu=\dim(\caW_\mu)$ and $d_\mu=\dim (\caS_\mu)$; let $P_\mu$ be the projector onto $\caH_\mu$; then we have $\tr(P_\mu)=\dim(\caH_\mu)=d_\mu D_\mu$. The following theorem is proved in  Appendix~\ref{app:ComputeOmegaASG}.
\begin{theorem}\label{thm:OmegaASG}
For $n\geq 3$ we have	
	\begin{align}
	\Omega_{\AS_n}^{\tilde{G}}&=\sum_{\mu} \frac{d_\mu}{D_\mu} P_\mu, \label{eq:OmegaASG} \\
	\beta(\Omega_{\AS_n}^{\tilde{G}})&=\frac{1}{n+1}, \quad \nu(\Omega_{\AS_n}^{\tilde{G}})=\frac{n}{n+1}. \label{eq:nuOmegaASG}
	\end{align}
	To verify the antisymmetric basis state $\ket{\AS_n}$ within infidelity $\epsilon$ and significance level $\delta$, the number of tests required reads
	\begin{equation}\label{eq:NumberTestASG}
	N\big(\epsilon,\delta,\Omega_{\AS_n}^{\tilde{G}}\big) \approx \frac{n+1}{n\epsilon}\ln\delta^{-1}.
	\end{equation}
\end{theorem}
Equation \eref{eq:nuOmegaASG} in Theorem~\ref{thm:OmegaASG} follows from \eref{eq:OmegaASG} and \lref{lem:DimRatio} below, which imply that  the second largest eigenvalue of $\Omega_{\AS_n}^{\tilde{G}}$ is $d_\mu/D_\mu$ with $\mu=(2,1,\ldots, 1)$. In this case, we have $d_\mu=n-1$ and $D_\mu=n^2-1$, which yields \eref{eq:nuOmegaASG}. Theorem~\ref{thm:OmegaASG} implies that our protocol associated with the  verification operator $\Omega_{\AS_n}^{\tilde{G}}$ is optimal for verifying the antisymmetric basis state $|\AS_n\>$ under LOCC. This
is the only optimal protocol known for
multipartite nonstabilizer states.  For quantum states with GME, it is extremely difficult to construct  optimal verification protocols under LOCC, and
such optimal protocols were known previously only  for GHZ states \cite{LiGHZ19} (optimal protocols for some other stabilizer states were constructed recently \cite{DangHZ20} after the initial posting of this paper).

A partition $\mu\vdash n$ is majorized by another partition $\mu'\vdash n$, denoted by $\mu\prec \mu'$, if
\begin{equation}
\sum_{j=1}^k\mu_j \leq \sum_{j=1}^k\mu_j'\quad \forall k=1,2\ldots, n.
\end{equation}
Note that the inequality is saturated when $k=n$.
The following lemma as  proved in  Appendix~\ref{app:ComputeOmegaASG} is very instructive to understanding the spectrum and spectral gap of the verification operator $\Omega_{\AS_n}^{\tilde{G}}$.
\begin{lemma}\label{lem:DimRatio}
	Suppose $\mu, \mu' \vdash n$  and $\mu\prec\mu'$; then
	\begin{equation}\label{eq:DimRatio}
	\frac{D_\mu}{d_\mu}\leq \frac{D_{\mu'}}{d_{\mu'}}.
	\end{equation}
\end{lemma}

\subsection{Efficient certification of GME}
A multipartite pure state is  genuinely multipartite entangled if it is not separable across every bipartition.
According to Ref.~\cite{Guhne09}, a quantum state $\rho$ is genuinely multipartite entangled if its fidelity with some
multipartite entangled state $|\Psi\>$ is larger than $C_\Psi$, where $C_\Psi$ is the square of the maximum  Schmidt coefficient of $|\Psi\>$ maximized over  all  bipartitions.
Note that $C_{\Psi}$ equals $1/n$ when $|\Psi\>$ is the antisymmetric basis state $|\AS_n\>$.
Thus a state $\rho$ is genuinely multipartite entangled if $\tr(\rho|\AS_n\>\<\AS_n|)>1/n$.
Given a verification strategy $\Omega$ for $|\AS_n\>$, to certify the GME of the antisymmetric basis state with significance level $\delta$, the number of tests is determined by \eref{eq:NumberTest} with $\epsilon=(n-1)/n$. If $\Omega$ is the optimal local strategy with $\nu(\Omega)=n/(n+1)$ (the strategy $\Omega_{\AS_n}^{\tilde{G}}$ constructed in Sec.~\ref{sec:OptimalASn} for example), then this number reads
\begin{equation}
N_\mathrm{E}=\biggl\lceil\frac{ \ln \delta}{\ln2-\ln(n+1)}\biggr\rceil,
\end{equation}
which decreases monotonically with $n$.
We have  $N_\mathrm{E}=1$ when $n\geq2\delta^{-1}-1$, so the GME of the antisymmetric basis state
can be certified with any given significance level using only one test when the number $n$ of particles is large enough. Previously, single-copy certification of GME is known only for GHZ states  \cite{LiGHZ19} and qudit stabilizer states \cite{ZhuEVQPSlong19}.
The current result is of special interest because it may shed light on the certification of GME of other nonstabilizer states.

\section{Comparison with quantum state tomography}

Before concluding this paper, it is instructive to compare our verification protocols with traditional tomography \cite{Haff05,Haah17,Donnell16}; cf. Ref.~\cite{ZhuEVQPSlong19}.
First, they have different assumptions. In quantum state tomography, it is usually assumed that the states  prepared in different runs are identical.
However, this assumption is difficult to guarantee in many scenarios of practical interest.
In quantum state verification, we can drop this assumption and thus draw a stronger conclusion \cite{ZhuEVQPSshort19,ZhuEVQPSlong19}.

Second, the two approaches address different tasks and have different goals. Quantum state tomography aims to determine the density matrix of an unknown quantum state completely.
That is why the resource required grows exponentially with the number $n$ of qubits (qudits),
given that the system size increases exponentially with $n$.
By contrast, the aim of quantum state verification is to determine whether the states prepared are sufficiently close to the target state on average.
If these states are far from the target state, then they will fail the tests quickly, so we can avoid false positive conclusion with high probability, but we can get little information about the true state in this case.
In a word, quantum state verification tries to extract the key information---the fidelity with the target state---as efficiently as possible. It is a powerful tool in many scenarios of practical interest in which quantum state tomography is too resource consuming to apply. It cannot replace  tomography completely, but is  a useful addition to the traditional tomographic approaches.
In practice, the choice of a specific method depends on the specific task and goal under consideration.

To see the inefficiency of tomography, here we review the number of copies (tests) required when tomography is employed to estimate an unknown quantum state within a  given precision as quantified by the trace distance $\epsilon_1$ or the infidelity $\epsilon_2$.
Recall that the fidelity between two density matrices $\rho$ and $\sigma$ is defined as
\begin{equation}
F(\rho,\sigma):=\left[ \tr\left(\sqrt{\sqrt{\rho}\sigma\sqrt{\rho}}\right) \right]^2.
\end{equation}
To clarify the efficiency  limit of the traditional approach, here we consider quantum state tomography with optimal collective measurements (the most general measurements allowed by quantum mechanics). The efficiency can only decrease if only individual local measurements are accessible.
Suppose $\sigma$ is an unknown $D$-dimensional quantum state, and $\rho$ is our estimator constructed using quantum state tomography.
If we want to achieve precision $\epsilon_1$ in trace distance, that is, $\frac{1}{2}\|\rho-\sigma\|_1\leq \epsilon_1$ (with constant probability close to 1); then at least $\Omega(D^2/\epsilon_1^2)$ copies are required according to Ref.~\cite{Haah17}.
In addition, Ref.~\cite{Donnell16} shows that $O(D^2/\epsilon_1^2)$ copies are sufficient to accomplish this task.
On the other hand, if we want to achieve infidelity $\epsilon_2$, that is, $1-F(\rho,\sigma)\leq \epsilon_2$, then $\Omega(D^2/\epsilon_2)$ copies are necessary, while
$O(D^2/\epsilon_2)\ln(D/\epsilon_2)$ copies are sufficient according to Ref.~\cite{Haah17}.

Next, we devise a scheme based on tomography for determining whether a given unknown state $\sigma$ is sufficiently close to the $D$-dimensional target pure state $|\Psi\>\<\Psi|$, that is, whether the infidelity $\epsilon_\sigma:=1-\<\Psi|\sigma|\Psi\>$ is smaller than some threshold $\epsilon$.
First, we do tomography using $N$ copies of  $\sigma$  and obtain the estimator $\rho$. Second, we calculate $\epsilon_\rho:=1-\<\Psi|\rho|\Psi\>$.
In order to ensure the condition $\epsilon_\sigma\leq\epsilon$, it suffices to ensure the following condition
\begin{equation}
\frac{1}{2}\|\rho-\sigma\|_1+\epsilon_\rho\leq \epsilon,
\end{equation}
because $\<\Psi|\rho|\Psi\> \leq \<\Psi|\sigma|\Psi\>+\frac{1}{2}\|\rho-\sigma\|_1$ [see Eq.~(9.96) in Ref.~\cite{Wilde19} for example].
In the tomography procedure, we  require that our estimator $\rho$ satisfies the condition $\frac{1}{2}\|\rho-\sigma\|_1\leq \epsilon/2$ (with constant probability);
and finally we accept the state $\sigma$ if and only if $\epsilon_\rho\leq \epsilon/2$. Suppose optimal collective measurements are accessible; then this scheme requires $\Theta(D^2/\epsilon^2)$ copies.
If the target state is an $n$-qudit state, then $N=\Theta(d^{2n}/\epsilon^2)$ copies are necessary and sufficient.

In the previous sections we have shown that only $O(n/\epsilon)$, $O(\sqrt{n}/\epsilon)$, and $O(1/\epsilon)$ tests are required to verify the $n$-partite phased Dicke states, $W$ state, and antisymmetric basis state, respectively, within infidelity $\epsilon$. Compared with tomography, quantum state verification
can extract the key information---the fidelity with the target state---exponentially more efficiently.
Nevertheless, it should be pointed out again that  the two
approaches rely on different assumptions and have different scopes of applications. Hence the above comparison cannot be completely fair.

\section{Summary}
Motivated by practical applications, we proposed several efficient protocols
for verifying general  phased Dicke states, including $W$ states and qudit Dicke states.
Our protocols only require adaptive local projective measurements, which are as simple as one can expect and are quite appealing in practice. To verify any $n$-qudit phased Dicke state within infidelity $\epsilon$ and significance level $\delta$, the number of tests required is only  $O(n\epsilon^{-1}\ln\delta^{-1})$, which is exponentially more efficient than previous approaches based on quantum state tomography and direct fidelity estimation. In addition, this number can be further reduced to $O(\sqrt{n}\,\epsilon^{-1}\ln\delta^{-1})$ for $W$ states.
One of our protocols for the three-qubit $W$ state is nearly optimal for both nonadversarial and adversarial scenarios, and it can also be applied to fidelity estimation.
Moreover, we constructed an optimal protocol for verifying the antisymmetric basis state; the number of tests required decreases monotonically with the number $n$ of particles. By virtue of this protocol, the GME of the  antisymmetric basis state can be certified with any given significance level using
 only one test when $n$ is sufficiently large.
In this way, our work provides powerful tools for characterizing and verifying various phased Dicke states. In the course of study,
we introduced several methods for  improving the efficiency of a given verification strategy, which are useful to studying quantum  verification in general. In addition, our work highlights the significance of graph theory and representation theory in studying quantum verification, which is of interest to many researchers beyond quantum information science.

\section*{Acknowledgments}
Z.L. thanks Jiahao Li for helpful discussion on the proof of  Proposition~\ref{pro:LimSqrt(n)h(n)}. H.Z. is grateful to Prof.~Eiichi Bannai for stimulating discussion on the spectrum of the transposition graph.
This work is supported by  the National Natural Science Foundation of China (Grant No.~11875110) and  Shanghai Municipal Science and Technology Major Project (Grant No.~2019SHZDZX01). J.S. acknowledges support by the Beijing Institute of Technology Research Fund Program for Young Scholars and the National Natural Science Foundation of China (Grant No.~11805010).

\onecolumngrid
\appendix
\section{\label{app:LemmaTwoTestProof}Proof of \lref{lem:2TestStrategy}}
\begin{proof}
Note that $|\Psi\>$ is an eigenstate of $P_1$ and $P_2$ with eigenvalue 1 by assumption.
Without loss of generality, we can assume that $P_1$ has rank $l+1$ and $P_2$ has rank $h+1$ with $h\leq l$. Then we can find two sets of orthonormal  states $\{|\phi_j\>\}_{j=1}^l$ and $\{|\varphi_k\>\}_{k=1}^h$ such that
\begin{equation}\label{eq:qk}
P_1=|\Psi\>\<\Psi|+\sum_{j=1}^l|\phi_j\>\<\phi_j|,       \qquad
P_2=|\Psi\>\<\Psi|+\sum_{k=1}^h|\varphi_k\>\<\varphi_k|, \qquad
|\<\phi_j|\varphi_k\>|^2=q_k\delta_{jk},
\end{equation}
where the overlaps $q_k$ are arranged in decreasing order, that is, $1\geq q_1\geq q_2\geq\cdots q_h\geq0$. As a consequence, we have
\begin{align}
\bar{P}_1=&\sum_{j=1}^l|\phi_j\>\<\phi_j|,       \qquad
\bar{P}_2=\sum_{k=1}^h|\varphi_k\>\<\varphi_k|,\\
q:=&\left\|\bar{P}_1\bar{P}_2 \bar{P}_1\right\|
=\Bigg\|\bigg(\sum_{j=1}^l|\phi_j\>\<\phi_j|\bigg) \bigg(\sum_{k=1}^h|\varphi_k\>\<\varphi_k|\bigg) \bigg(\sum_{j=1}^l|\phi_j\>\<\phi_j|\bigg)\Bigg\|
=\Bigg\| \sum_{k=1}^h q_k |\phi_k\>\<\phi_k| \Bigg\|
=q_1,\\
\max_{|\phi\>\in \supp(\bar{P}_1)} \<\phi|P_2|\phi\>=&\max_{\sum_{j=1}^l|c_j|^2=1}\sum_{j,k=1}^l c_j^* c_k\<\phi_j |P_2|\phi_k\>=\max_{\sum_{j=1}^l|c_j|^2=1} \sum_{j=1}^h q_j |c_j|^2=q_1 =\left\|\bar{P}_1\bar{P}_2 \bar{P}_1\right\|=q.
\end{align}

In addition, the verification operator $\Omega$ can be expressed as follows,
\begin{equation}
\Omega=pP_1+(1-p)P_2=|\Psi\>\<\Psi|+\sum_{j=1}^h \big[ p|\phi_j\>\<\phi_j|+(1-p)|\varphi_j\>\<\varphi_j| \big]+p\sum_{k=h+1}^l|\phi_k\>\<\phi_k|.
\end{equation}
So the second largest eigenvalue of $\Omega$ reads
\begin{align}
\lambda_2(\Omega)&=\bigl\|\Omega-|\Psi\>\<\Psi|\bigr\|
 =\max_{1\leq j\leq h} \bigl\|p|\phi_j\>\<\phi_j|+(1-p)|\varphi_j\>\<\varphi_j|\bigr\| =\max_{1\leq j\leq h} \frac{1}{2}\Big[1+\sqrt{(2p-1)^2+4p(1-p)q_j} \,\Big]
 \nonumber\\
 &=\frac{1}{2}\Big[1+\sqrt{(2p-1)^2+4p(1-p)q_1}\, \Big]=\frac{1}{2}\Big[1+\sqrt{(2p-1)^2+4p(1-p)q}\, \Big] \nonumber\\
&= \frac{1}{2}\Big[1+\sqrt{4(1-q)p^2-4(1-q)p+1}\Big]
\geq  \frac{1+\sqrt{q}}{2}.
\end{align}
If $q<1$, then   the lower bound  is saturated iff $p=1/2$, in which case we have $\Omega=(P_1+P_2)/2$.
Therefore, the spectral gap satisfies $\nu(\Omega)\leq(1-\sqrt{q})/2$, and the upper bound is saturated iff $p=1/2$, which confirms \lref{lem:2TestStrategy}.
\end{proof}

\section{\label{app:TheoDickeProof}Proof of \tref{thm:Dicke}}
\begin{proof}
The verification operator $\Omega_\bfk$ can be expressed as
\begin{align}\label{eq:decomOmegaD}
\Omega_\bfk
&= \binom{n}{2}^{-1}\sum_{i<j} \sum_{s=0}^g \bcaZ_{i,j}(\bfk_{ss})       \otimes \left[(\ket{s}\bra{s})^{\otimes2}\right]_{i,j}
      +\binom{n}{2}^{-1}\sum_{i<j}\sum_{s<t}\bcaZ_{i,j}(\bfk_{st}) \otimes
      \left[(T_{s,t}^+)^{\otimes 2}+(T_{s,t}^-)^{\otimes 2}\right]_{i,j}\nonumber\\
&=\frac{1}{n(n-1)} \Bigg(\sum_{s=0}^r k_s^2-n\Bigg) \mathcal{Z}(\bfk)
      +\frac{2}{n(n-1)}\sum_{i<j}\sum_{s<t}   \bcaZ_{i,j}(\bfk_{st})
      \otimes \left[\bigl(\ket{\psi_{s,t}^+}\bra{\psi_{s,t}^+}\bigr)_{i,j}+\bigl(\ket{\varphi_{s,t}^+}\bra{\varphi_{s,t}^+}\bigr)_{i,j}\right] \nonumber\\
&= \frac{1}{n(n-1)}\bigg[M_1+\sum_{s<t}M_{(s,t)}\bigg]\,.
\end{align}
Here $\ket{\psi_{s,t}^+}=\frac{1}{\sqrt{2}}(\ket{s}\ket{t}+\ket{t}\ket{s})$, $\ket{\varphi_{s,t}^+}=\frac{1}{\sqrt{2}}(\ket{s}\ket{s}+\ket{t}\ket{t})$,
$\caZ(\bfk)=\sum_{u\in B(\bfk) }\ket{u}\bra{u}$, and
\begin{align}
M_1&:=\Bigg(\sum_{s=0}^r k_s^2-n+\sum_{s<t} k_s k_t \Bigg)\caZ(\bfk)+\sum_{\substack{u,v\in B(\bfk)\\ u\sim v }}\ket{u}\bra{v}
    =\frac{1}{2}\Bigg(n^2-2n+\sum_{s=0}^r k_s^2 \Bigg)\caZ(\bfk)+\sum_{u,v\in B(\bfk) }A_{u v}\ket{u}\bra{v}\,,\label{eq:M1}\\
M_{(s,t)}
   &:= \frac{k_s(k_s+1)}{2}\sum_{u\in B(\bfk^s_{t}) }\ket{u}\bra{u}
     +\frac{k_t(k_t+1)}{2}\sum_{v\in B(\bfk^t_{s}) }\ket{v}\bra{v}
     +\sum_{\substack{u\in B(\bfk^s_{t})\\v\in B(\bfk^t_{s})\\ u\sim v}}\bigr(\ket{u}\bra{v}+\ket{v}\bra{u}\bigr),
\label{eq:M(s,t)}
\end{align}
where the notation $u\sim v$ means  $u_j\neq v_j$ for exactly two values of $j$.
The coefficient matrix $(A_{uv})$ for $u,v\in B(\bfk)$ happens to be the
adjacency matrix $A(\bfk)$ of the transposition graph $G(\bfk)$ \cite{Chase1973} explained in  Appendix~\ref{app:SpecGraph}.
Note that $M_1$ and all $M_{(s,t)}$ (with $s,t=0,1,\dots,r$ and  $s<t$) are Hermitian and have mutually orthogonal supports, so all of them are
positive semidefinite given that $\Omega_\bfk$ is positive semidefinite by construction.

According to \lref{lem:GraphSpect} in Appendix~\ref{app:SpecGraph}, the maximum eigenvalue of $A(\bfk)$ is $d=\big(n^2-\sum_{s=0}^r k_s^2\big)/2$ with multiplicity~1,
and the second largest eigenvalue of $A(\bfk)$ is equal to $d-n$.
Therefore, the two largest eigenvalues of $M_1$ read
\begin{equation}\label{eq:M1lambda}
\lambda_1(M_1)=n(n-1), \qquad \lambda_2(M_1)=n(n-1)-n=n(n-2).
\end{equation}
In addition, direct calculations show that $M_{(s,t)}$ has an eigenstate
\begin{align}
\ket{\Theta_{s,t}}&=\frac{1}{\sqrt{k_s(k_s+1)+k_t(k_t+1)}}\Big[ \sqrt{k_s(k_s+1)}\,  \big|D(\bfk^s_{t})\big\>
 +\sqrt{k_t(k_t+1)}\,  \big|D(\bfk^t_{s})\big\>  \Big]\nonumber\\
 &=\sqrt{\frac{\prod_{j=0}^r k_j!}{ k_s k_t[k_s(k_s+1)+k_t(k_t+1)](n!)}}\,\Biggl[k_s(k_s+1)\sum_{u\in B(\bfk^s_{t})} |u\>+ k_t(k_t+1)\sum_{u\in B(\bfk^t_{s})} |u\>\Biggr],
\label{eq:defphiD}
\end{align}
with eigenvalue
\begin{equation}\label{eq:M(s,t)lambda}
\lambda_1\big(M_{(s,t)}\big)=\frac{k_s(k_s+1)}{2}+\frac{k_t(k_t+1)}{2}\,.
\end{equation}
According to the Perron-Frobenius theorem (see Chapter~8 in Ref.~\cite{Meyer2000} for example), this is the largest  eigenvalue of $M_{(s,t)}$, given that $M_{(s,t)}$ is irreducible in the subspace spanned by $\ket{u}$ with $u\in B(\bfk^s_{t})\cup B(\bfk^t_{s})$, that is, the graph corresponding to the third term of
$M_{(s,t)}$ in \eref{eq:M(s,t)} is connected.
In conjunction with Eqs.~\eqref{eq:decomOmegaD} and \eqref{eq:M1lambda}, we can deduce
the second largest eigenvalue and  spectral gap of $\Omega_\bfk$, with the result
\begin{align}
\lambda_2(\Omega_\bfk)&=\max\left\{\frac{\lambda_2(M_1)}{n(n-1)},\,\max_{s<t}\frac{ \lambda_1\big(M_{(s,t)}\big) }{n(n-1)} \right\}
=\max\left\{\frac{n-2}{n-1}, \frac{ k_0(k_0+1)+k_1(k_1+1) }{2n(n-1)}   \right\},\\
\nu(\Omega_\bfk)&=1-\lambda_2(\Omega_\bfk)=\min\left\{\frac{1}{n-1},1-\frac{ k_0(k_0+1)+k_1(k_1+1) }{2n(n-1)} \right\}.
\label{eq:gapD}
\end{align}
When $n\ge 4$, the above equations reduce to
\begin{equation}
  \lambda_2(\Omega_\bfk)=\frac{n-2}{n-1}\,,\qquad\qquad \nu(\Omega_\bfk)=\frac{1}{n-1}\,,
\end{equation}
which confirms \eref{eq:nuOmegaD}. When $n=3$, \eref{eq:nuOmegaD} can be verified directly by virtue of \eref{eq:gapD}.

Equation \eqref{eq:NumberTestDicke} follows from \eqsref{eq:NumberTest}{eq:nuOmegaD}. This observation completes the proof of \tref{thm:Dicke}.
\end{proof}

\section{\label{app:ComputeOmegaD^G}Determination of  $\Omega_\bfk^G$ and proof of \eref{eq:nuOmegaD2}}
Denote by $H$ the group of all unitary transformations of the form $U^{\otimes n}$, where  $U$ is diagonal in the computational basis. According to \eqsref{eq:OmegaSym}{eq:decomOmegaD}, we have
\begin{align}\label{eq:decomOmegaD^G}
\Omega_\bfk^G=\Omega_\bfk^H=\frac{M_1^H+\sum_{s<t}M_{(s,t)}^H}{n(n-1)},
\end{align}
where
\begin{equation}\label{eq:M^H}
M_1^H=M_1, \quad   M_{(s,t)}^H =
\frac{k_s(k_s+1)}{2}\sum_{u\in B(\bfk^s_{t}) }\ket{u}\bra{u}
+\frac{k_t(k_t+1)}{2}\sum_{v\in B(\bfk^t_{s}) }\ket{v}\bra{v}.
\end{equation}
Equation \eqref{eq:M^H} follows from \eqsref{eq:M1}{eq:M(s,t)} as well as the fact that $(|u\>\<v|)^H=|u\>\<v|$ if $u$ and $v$ can be turned into each other by a permutation, while $(|u\>\<v|)^H=0$ otherwise.

Note that $M_1$ and all $M^H_{(s,t)}$ (with $s,t=0,1,\dots,r$ and $s<t$) are positive semidefinite and have mutually orthogonal supports.
The largest eigenvalue of $M^H_{(s,t)}$ reads $\lambda_1\big(M_{(s,t)}^H\big)=k_s(k_s+1)/2$.
In conjunction with \eqssref{eq:decomOmegaD^G}{eq:M^H}{eq:M1lambda}, we can deduce
the second largest eigenvalue and the spectral gap of $\Omega_\bfk^G$, with the result
\begin{equation}
\lambda_2\big(\Omega^G_\bfk\big)
=\max\left\{\frac{\lambda_2(M_1^H)}{n(n-1)},\,\max_{s<t}\frac{ \lambda_1\big(M_{(s,t)}^H\big) }{n(n-1)} \right\}=
\max\left\{\frac{n-2}{n-1}, \frac{ k_0(k_0+1) }{2n(n-1)}   \right\}= \frac{n-2}{n-1}, \qquad\quad
\nu\big(\Omega^G_\bfk\big)
=\frac{1}{n-1},
\end{equation}
which confirms \eref{eq:nuOmegaD2}.

\section{\label{app:derive q}Proof of \eref{eq:q}}
\begin{proof}
Note that each ket $|\zeta\>$ in the support of the test projector	$P_1$ in \eref{eq:P1}  has the following form
\begin{equation}\label{eq:zetaState}
|\zeta\>=a_1|10\dots00\>+a_2|01\dots00\>+\cdots+a_n|00\dots01\>,
\end{equation}
where $a_1,a_2,\dots,a_n$ are complex numbers that satisfy the normalization condition $\sum_{j=1}^n |a_j|^2=1$.
In conjunction with  \lref{lem:2TestStrategy} and Eqs.~\eqref{eq:alphax}-\eqref{eq:P2}, we can deduce that
\begin{align}\label{eq:q expansion}
q=&\|\bar{P}_1\bar{P}_2 \bar{P}_1\|=\max_{\substack{\<W_n|\zeta\>=0\\ \<\zeta|P_1|\zeta\>=1}}\<\zeta|P_2|\zeta\>
     =\frac{1}{2^{n-1}}  \max_{\substack{\sum_{j} a_j=0\\\sum_{j} |a_j|^2=1}} \sum_{x\in\{0,1\}^{n-1}}
                 \frac{\left|a_n+ (n-1-2|x|)\sum_{j=1}^{n-1}(-1)^{x_j}a_j \right|^2}{1+\left(n-1-2|x|\right)^2}
                     \nonumber\\
  =&\frac{1}{2^{n-1}}  \max_{\substack{\sum_{j} a_j=0\\\sum_{j} |a_j|^2=1}} \sum_{k=0}^{n-1} \frac{1}{1+\left(n-1-2k\right)^2}
                         \sum_{\substack{x\in\{0,1\}^{n-1}\\|x|=k}}
                         \left|a_n+ (n-1-2k)\sum_{j=1}^{n-1}(-1)^{x_j}a_j \right|^2,
\end{align}
where $|x|$ denotes the Hamming weight of $x$.
When  $x\in\{0,1\}^{n-1}$ and $a_1,a_2,\dots,a_n$ satisfy the conditions $\sum_{j} a_j=0$ and  $\sum_{j}|a_j|^2=1$,  we can derive the following equalities,
\begin{align}
&\sum_{|x|=k}
                         \left|a_n+ (n-1-2k)\sum_{j=1}^{n-1}(-1)^{x_j}a_j \right|^2
                                     \nonumber\\
&=\sum_{|x|=k} |a_n|^2
                        +2\rmRe\!\left[(n-1-2k) a_n^* \sum_{|x|=k}\sum_{j=1}^{n-1}(-1)^{x_j}a_j\right]
                        +(n-1-2k)^2 \sum_{|x|=k} \left|\sum_{j=1}^{n-1}(-1)^{x_j}a_j \right|^2
                                     \nonumber\\
&=\frac{4k(n-1-k)(n-1-2k)^2}{(n-1)(n-2)}\binom{n-1}{k}
                        +\binom{n-1}{k}\left\{1+(n-1-2k)^2\left[1-\frac{2}{n-1}-\frac{8k(n-1-k)}{(n-1)(n-2)}\right]\right\}|a_n|^2\label{eq:|.| expansion}
\end{align}
for $k=0,1,\dots,n-1$, where $a_j^*$ denotes the complex conjugate of $a_j$.
The last equality in \eref{eq:|.| expansion} follows from the two equations below,
\begin{align}
\sum_{|x|=k}\sum_{j=1}^{n-1}(-1)^{x_j}a_j
&=\frac{n-1-2k}{n-1}\binom{n-1}{k}\sum_{j=1}^{n-1}a_j=-\frac{n-1-2k}{n-1}\binom{n-1}{k}a_n, \\
\sum_{|x|=k} \Biggl|\sum_{j=1}^{n-1}(-1)^{x_j}a_j \Biggr|^2
&= \binom{n-1}{k} \left[ \sum_{j=1}^{n-1}|a_j|^2
+\frac{(n-1)(n-2)-4k(n-1-k)}{(n-1)(n-2)} \sum_{\substack{ i,j=1,\dots,n-1 \\i\ne j}} a_i a_j^*\right]  \nonumber\\
&=\binom{n-1}{k}\left[ \frac{4k(n-1-k)}{(n-1)(n-2)}+\frac{(n-1)(n-2)-8k(n-1-k)}{(n-1)(n-2)}|a_n|^2\right]. \label{eq:AltSumSq}
\end{align}
In deriving the second equality in \eref{eq:AltSumSq}, we have employed  the following facts
\begin{equation}
\sum_{j=1}^{n-1}|a_j|^2=1-|a_n|^2,\quad\sum_{\substack{ i,j=1,\dots,n-1 \\i\ne j}} a_i a_j^*=\left|\sum_{ i=1}^{n-1} a_i\right|^2 -\sum_{j=1}^{n-1}|a_j|^2=|-a_n|^2-(1-|a_n|^2)=2|a_n|^2-1.
\end{equation}

Now by plugging \eref{eq:|.| expansion} into \eref{eq:q expansion} we obtain
\begin{align}\label{eq:qan}
q= c_1(n)+\max_{\substack{\sum_{j} a_j=0\\\sum_{j} |a_j|^2=1}} c_2(n)|a_n|^2,
\end{align}
where the coefficients $c_1(n)$ and $c_2(n)$ read
\begin{align}
c_1(n)
:=&\frac{1}{2^{n-1}}\sum_{k=0}^{n-1}\frac{\binom{n-1}{k}4k(n-1-k)(n-1-2k)^2}{(n-1)(n-2)[1+(n-1-2k)^2]}=\frac{1}{2^{n-3}}\sum_{k=1}^{n-2}\frac{\binom{n-3}{k-1}(n-1-2k)^2}{[1+(n-1-2k)^2]}\nonumber\\=&\frac{1}{2^{n-3}}\sum_{k=0}^{n-3}\frac{\binom{n-3}{k}[1+(n-3-2k)^2-1]}{[1+(n-3-2k)^2]}
=1-h(n-3),\\
c_2(n)
:=&\frac{1}{2^{n-1}}\sum_{k=0}^{n-1}\frac{\binom{n-1}{k}}{1+(n-1-2k)^2}\left\{1+(n-1-2k)^2\left[1-\frac{2}{n-1}-\frac{8k(n-1-k)}{(n-1)(n-2)}\right]\right\} \nonumber\\
=&\frac{1}{2^{n-1}}\sum_{k=0}^{n-1}\binom{n-1}{k}-\frac{2}{n-1}\frac{1}{2^{n-1}}\sum_{k=0}^{n-1}\frac{\binom{n-1}{k}(n-1-2k)^2}{1+(n-1-2k)^2}-2c_1(n)\nonumber\\
=&1-\frac{2}{n-1}[1-h(n-1)]-2[1-h(n-3)]
=\frac{2}{n-1}h(n-1)+2h(n-3)-\frac{n+1}{n-1},
\end{align}
and $h(n)$ is defined in \eref{eq:hn}.

If  $n=3$, then $c_1(n)=0$ and $c_2(n)=\frac{3}{5}>0$, so the maximum in \eref{eq:qan} is attained when
$a_1=a_2=-\frac{1}{\sqrt{6}}$  and $a_3=\sqrt{\frac{2}{3}}$, in which case  we have
\begin{equation}
q=c_1(3)+\frac{2}{3}c_2(3)=\frac{2}{5},
\end{equation}
which confirms \eref{eq:q} in the case $n=3$.

If $n=4,5$, then  $c_2(n)<0$ by direct calculation.  If $n\geq 6$, then
\begin{equation}
c_2(n)=\frac{2}{n-1}h(n-1)+2h(n-3)-\frac{n+1}{n-1}< \frac{1}{n-1}+1-\frac{n+1}{n-1}=-\frac{1}{n-1}<0,
\end{equation}
where the first inequality follows from the fact that $h(n)<1/2$ for $n\geq3$, which is easy to prove.
Therefore, $c_2(n)<0$ for $n\geq4$. In this case, the  maximum in \eref{eq:qan} is attained when
 $a_1=-a_2=\frac{1}{\sqrt{2}}$ and $a_j=0$ for $j=3,4,\dots,n$, which yields
\begin{equation}
q=c_1(n)=1-h(n-3)
\end{equation}
and confirms \eref{eq:q}.
\end{proof}

\section{\label{app:h(n)property}Proofs of Propositions~\ref{pro:sqrt{n}h(n)Monot} and \ref{pro:LimSqrt(n)h(n)}}
\begin{proof}[Proof of Proposition~\ref{pro:sqrt{n}h(n)Monot}]
To prove Proposition~\ref{pro:sqrt{n}h(n)Monot}, it suffices to prove that $\sqrt{n+2}\,h(n+2)>\sqrt{n}\,h(n)$ for each integer $n\geq0$.
When $n=0$, the inequality is obvious;
when $n\geq1$, the inequality can be proved as follows,
\begin{align}
&\frac{2^{n+2}}{\sqrt{n}}\left[\sqrt{n+2}\,h(n+2)-\sqrt{n}\,h(n)\right]
=\sqrt{\frac{n+2}{n}}\sum_{j=0}^{n+2}\frac{\binom{n+2}{j}}{1+(n+2-2j)^2}
-\sum_{j=0}^{n}\frac{4\binom{n}{j}}{1+(n-2j)^2}
\nonumber\\
&\quad> \frac{n+2}{n+1}\sum_{j=0}^{n+2}\frac{\binom{n+2}{j}}{1+(n+2-2j)^2}
-\sum_{j=0}^{n}\frac{4\binom{n}{j}}{1+(n-2j)^2}\nonumber= \frac{n+2}{n+1}\sum_{k=-1}^{n+1}\frac{\binom{n+2}{k+1}}{1+(n-2k)^2}
-\sum_{j=0}^{n}\frac{4\binom{n}{j}}{1+(n-2j)^2}
\nonumber\\
&\quad > \sum_{j=0}^{n}\frac{1}{1+(n-2j)^2}\left[\frac{n+2}{n+1}\binom{n+2}{j+1}-4\binom{n}{j}\right]\geq0.
\end{align}
Here the first inequality holds because $\sqrt{\frac{n+2}{n}}>\frac{n+2}{n+1}$, and
the last inequality holds because
\begin{equation}
\frac{n+2}{n+1}\binom{n+2}{j+1}-4\binom{n}{j}
=\left[\frac{(n+2)^2}{(j+1)(n+1-j)}-4\right]\binom{n}{j}
\geq0,\qquad j=0,1,\dots,n.
\end{equation}
Therefore, $\sqrt{n}\,h(n)$ is strictly monotonically increasing in $n$ for odd $n$ and even $n$, respectively.
\end{proof}

\begin{proof}[Proof of Proposition~\ref{pro:LimSqrt(n)h(n)}]First, \eref{eq:LimSqrt(n)h(n)Even} in Proposition~\ref{pro:LimSqrt(n)h(n)} can be derived as follows,
	\begin{align}
	& \lim_{n \to +\infty} \sqrt{2n}\,h(2n)
	=\lim_{n \to +\infty} \frac{\sqrt{2n}}{2^{2n}}\sum_{j=0}^{2n}\frac{\binom{2n}{j}}{1+(2n-2j)^2} = \left[\lim_{n \to +\infty} \frac{\sqrt{2n}}{2^{2n}}\binom{2n}{n}\right]
	\left[\lim_{n\to+\infty} \sum_{j=0}^{2n}\frac{\binom{2n}{n}^{-1}\binom{2n}{j}}{1+(2n-2j)^2}
	\right] \nonumber\\
	&=\sqrt{\frac{2}{\pi}}
	\left[\lim_{n\to+\infty} \sum_{j=0}^{2n}\frac{1}{1+(2n-2j)^2}
	\right]= \sqrt{\frac{2}{\pi}}
	\left[1+ \lim_{n \to +\infty} \sum_{k=1}^{n}\frac{2}{1+(2k)^2}\right] \nonumber\\
	&= \sqrt{\frac{2}{\pi}} \left[1+ \frac{\pi}{2}\coth\left(\frac{\pi}{2}\right)-1 \right]
	= \sqrt{\frac{\pi}{2}}  \coth\Bigl(\frac{\pi}{2}\Bigr)
	\approx 1.37,  \label{eq:limEvenProof0}
	\end{align}
	where the third equality follows from \eqsref{eq:limEvenProof1}{eq:limEvenProof2} below, and the	
	fifth equality is a corollary of \eref{eq:limEvenProof3} below,
	\begin{align}
	&\lim_{n \to +\infty} \frac{\sqrt{2n}}{2^{2n}}\binom{2n}{n}
	= \lim_{n \to +\infty} \frac{\sqrt{2n}(2n)!}{2^{2n}(n!)^2}
	= \sqrt{\frac{2}{\pi}}\, , \label{eq:limEvenProof1}  \\
	&\lim_{n\to+\infty} \sum_{j=0}^{2n}\frac{\binom{2n}{n}^{-1}\binom{2n}{j}}{1+(2n-2j)^2}=\lim_{n\to+\infty} \sum_{j=0}^{2n}\frac{1}{1+(2n-2j)^2}, \label{eq:limEvenProof2}\\
	&\lim_{n \to +\infty} \sum_{k=1}^{n}\frac{2}{1+(2k)^2}
	=\frac{\rmi}{2} \lim_{n \to +\infty} \sum_{k=1}^{n} \Big(\frac{1}{\rmi/2-k}+\frac{1}{\rmi/2+k}\Big)
	=\frac{\rmi}{2}\Big[\pi\cot\Bigl(\frac{\rmi\pi}{2}\Bigr)+2\rmi\Big]
	=\frac{\pi}{2}\coth\Bigl(\frac{\pi}{2}\Bigr)-1. \label{eq:limEvenProof3}
	\end{align}
	The second equality in \eref{eq:limEvenProof1} follows from the Wallis formula [see Eq.~(1) in Ref.~\cite{Piros03} for example] or the Stirling formula;
 the second equality in  \eref{eq:limEvenProof3} follows from  Theorem~6.12 in Ref.~\cite{Ullri08}.

To prove \eref{eq:limEvenProof2},
note that the left-hand side in \eref{eq:limEvenProof2} cannot be larger than the right-hand side thanks to the inequality $\binom{2n}{j}\leq \binom{2n}{n}$. To complete the proof, it suffices to prove the opposite inequality, which can be derived as follows.
\begin{align}
&\lim_{n\to+\infty} \sum_{j=0}^{2n}\frac{1}{1+(2n-2j)^2}-
\lim_{n\to+\infty} \sum_{j=0}^{2n}\frac{\binom{2n}{n}^{-1}\binom{2n}{j}}{1+(2n-2j)^2}
=2\lim_{n\to+\infty} \sum_{k=1}^{n}\frac{1-\binom{2n}{n}^{-1}\binom{2n}{n+k}}{1+4k^2}\nonumber\\
&=
2\lim_{n\to+\infty} \sum_{k=1}^{\lceil n^{2/3}\rceil}\frac{1-\binom{2n}{n}^{-1}\binom{2n}{n+k}}{1+4k^2}+
2\lim_{n\to+\infty} \sum_{\lceil n^{2/3}\rceil+1}^n\frac{1-\binom{2n}{n}^{-1}\binom{2n}{n+k}}{1+4k^2}
\nonumber\\
&\leq 2\lim_{n\to+\infty}
  \sum_{k=1}^{\lceil n^{2/3}\rceil}\frac{k^2}{n(1+4k^2)}+2\lim_{n\to+\infty} \sum_{\lceil n^{2/3}\rceil}^n\frac{1}{1+4k^2}\leq
 2\lim_{n\to+\infty}
  \frac{\lceil n^{2/3}\rceil}{4n}+2\lim_{n\to+\infty} \frac{n}{4n^{4/3}}=0,
\end{align}
where the first inequality is a consequence of the following equation,
\begin{equation}
 \binom{2n}{n}^{-1}\binom{2n}{n+k}=\frac{(n!)^2}{(n+k)!(n-k)!}=\frac{n(n-1)\cdots (n-k+1)}{(n+k)(n+k-1)\cdots (n+1)}\geq \left(\frac{n-k}{n}\right)^k\geq 1-\frac{k^2}{n},\quad k\in \{1,2,\dots,n\}.
\end{equation}

Next, \eref{eq:LimSqrt(n)h(n)Odd} in Proposition~\ref{pro:LimSqrt(n)h(n)} can be derived as follows,
	\begin{align}
& \lim_{n \to +\infty} \sqrt{2n+1}\,h(2n+1)
=\lim_{n \to +\infty} \frac{\sqrt{2n+1}}{2^{2n+1}}\sum_{j=0}^{2n+1}\frac{\binom{2n+1}{j}}{1+(2n+1-2j)^2} \nonumber\\
&= \left[\lim_{n \to +\infty} \frac{\sqrt{2n+1}}{2^{2n+1}}\binom{2n+1}{n}\right]
\left[\lim_{n\to+\infty} \sum_{j=0}^{2n+1}\frac{\binom{2n+1}{n}^{-1}\binom{2n+1}{j}}{1+(2n+1-2j)^2}
\right]= \sqrt{\frac{2}{\pi}}
\left[\lim_{n\to+\infty} \sum_{j=0}^{2n+1}\frac{1}{1+(2n+1-2j)^2}
\right]\nonumber\\
&= \sqrt{\frac{2}{\pi}}
\left[ \lim_{n \to +\infty} \sum_{k=0}^{n}\frac{2}{1+(2k+1)^2}\right]= \sqrt{\frac{2}{\pi}}\times \frac{\pi}{2}\tanh\left(\frac{\pi}{2}\right)
= \sqrt{\frac{\pi}{2}}  \tanh\Bigl(\frac{\pi}{2}\Bigr)
\approx 1.15,  \label{eq:limOddProof0}
\end{align}	
where the third equality follows from
\eqsref{eq:limOddProof1}{eq:limOddProof2} below, and
the fifth equality is a corollary of \eref{eq:limOddProof3} below.
	\begin{align}
	&\lim_{n \to +\infty} \frac{\sqrt{2n+1}}{2^{2n+1}}\binom{2n+1}{n}
	= \bigg[\lim_{n \to +\infty} \frac{\sqrt{2n}(2n)!}{2^{2n}(n!)^2}\bigg]    \bigg[\lim_{n \to +\infty} \sqrt{\frac{2n+1}{2n}} \frac{2n+1}{2(n+1)}  \bigg]
	= \sqrt{\frac{2}{\pi}},                                             \label{eq:limOddProof1}\\
&\lim_{n\to+\infty} \sum_{j=0}^{2n+1}\frac{\binom{2n+1}{n}^{-1}\binom{2n+1}{j}}{1+(2n+1-2j)^2}
= \lim_{n\to+\infty} \sum_{j=0}^{2n+1}\frac{1}{1+(2n+1-2j)^2},\label{eq:limOddProof2}\\			
	&\lim_{n \to +\infty} \sum_{k=0}^{n}\frac{2}{1+(2k+1)^2}
	=\bigg(\lim_{n \to +\infty} \sum_{k=1}^{2n+1}\frac{2}{1+k^2}\bigg)  -  \bigg[\lim_{n \to +\infty}  \sum_{k=1}^{n}\frac{2}{1+(2k)^2}\bigg] \nonumber       \\
	&\quad =\rmi \lim_{n \to +\infty} \sum_{k=1}^{n} \Big(\frac{1}{\rmi-k}+\frac{1}{\rmi+k}\Big) - \Big[\frac{\pi}{2}\coth\Bigl(\frac{\pi}{2}\Bigr)-1\Big]
	=\rmi\big[\pi\cot(\rmi\pi)+\rmi\big]-\frac{\pi}{2}\coth\Bigl(\frac{\pi}{2}\Bigr)+1
	=\frac{\pi}{2}\tanh\Bigl(\frac{\pi}{2}\Bigr). \label{eq:limOddProof3}
	\end{align}
	The second equality in \eref{eq:limOddProof1} follows from the Wallis formula [cf.~\eref{eq:limEvenProof1}];
	the second and third  equalities in \eref{eq:limOddProof3} follow from \eref{eq:limEvenProof3} above and Theorem~6.12 in Ref.~\cite{Ullri08}, respectively;
	\eref{eq:limOddProof2} can be proved in a similar way to \eref{eq:limEvenProof2}, given the following equation,
\begin{equation}
	\binom{2n+1}{n}^{-1}\binom{2n+1}{n+1+k}=\frac{n!(n+1)!}{(n+1+k)!(n-k)!}\geq \left(\frac{n-k}{n+1}\right)^k\geq 1-\frac{k^2+k}{n+1},\quad  k\in \{0,1,\dots,n\}.
	\end{equation}
\end{proof}

\section{\label{app:ProofEq:bound-nuW}Bounds for $\nu(\Omega_{W_n})$ and proof of \eref{eq:bound-nuW}}
To derive lower bounds for $\nu(\Omega_{W_n})$, we shall consider two cases depending on the parity of the qubit number $n$.
\begin{enumerate}
\item[1.] $n$ is an odd integer.

Direct calculation based on Eqs.~\eqref{eq:q}-\eqref{eq:nuW} shows that $\sqrt{n}\nu(\Omega_{W_n})>3/10$ for $3\leq n\leq33$. When $n\geq 35$, we have
\begin{equation}
\sqrt{n}\,\nu(\Omega_{W_n})
>\frac{1}{4}\sqrt{n-3}\,h(n-3)
\geq\frac{1}{4}\sqrt{32}\,h(32)>\frac{3}{10},
\end{equation}
where the first inequality follows from \eref{eq:nuW}, and
the second inequality follows from Proposition~\ref{pro:sqrt{n}h(n)Monot}. Therefore, $\sqrt{n}\nu(\Omega_{W_n})>3/10$ when $n$ is odd and $n\geq 3$, which implies the lower bound in \eref{eq:bound-nuW}.

\item[2.] $n$ is an even integer.

Direct calculation based on \eqsref{eq:hn}{eq:nuW} shows that  $\sqrt{n}\nu(\Omega_{W_n})>1/4$ for $4\leq n\leq42$. When $n\geq 44$, we have
\begin{equation}
\sqrt{n}\,\nu(\Omega_{W_n})
>\frac{1}{4}\sqrt{n-3}\,h(n-3)
\geq\frac{1}{4}\sqrt{41}\,h(41)>\frac{1}{4},
\end{equation}
where the first inequality follows from \eref{eq:nuW}, and
 the second inequality follows from Proposition~\ref{pro:sqrt{n}h(n)Monot}. Therefore, $\sqrt{n}\nu(\Omega_{W_n})>1/4$ when $n$ is even and $n\geq 4$,  which implies the lower bound in \eref{eq:bound-nuW} again.
\end{enumerate}
In conclusion, the lower bound in  \eref{eq:bound-nuW} holds for any integer $n$ that satisfies $n\geq 3$.

To derive upper bounds for $\nu(\Omega_{W_n})$, we also consider two cases depending on the parity of the qubit number $n$.
\begin{enumerate}
\item[1.] $n$ is an odd integer.

Direct calculation based on Eqs.~\eqref{eq:q}-\eqref{eq:nuW} shows that $\sqrt{n}\nu(\Omega_{W_n})<3/8$ for $3\leq n\leq45$ with $n\ne5$ and $\sqrt{n}\nu(\Omega_{W_n})<0.411<1/2$ when $n=5$. When $n\geq 47$, we have
\begin{equation}
\sqrt{n}\,\nu(\Omega_{W_n})
=\sqrt{n-3}\,h(n-3) \sqrt{\frac{n}{n-3}}  \frac{1-\sqrt{1-h(n-3)}}{2h(n-3)}
< \sqrt{\frac{\pi}{2}}  \coth\Bigl(\frac{\pi}{2}\Bigr)  \times  1.034  \times   0.265  < \frac{3}{8},
\end{equation}
where the first inequality follows from \eref{eq:hnbound2} and the following equations,
\begin{gather}
\sqrt{\frac{n}{n-3}}\leq \sqrt{\frac{47}{47-3}}<1.034,    \\
h(n-3)\leq \frac{1}{\sqrt{n-3}} \times \sqrt{\frac{\pi}{2}}  \coth\Bigl(\frac{\pi}{2}\Bigr)
\leq \sqrt{\frac{\pi}{2(47-3)}}  \coth\Bigl(\frac{\pi}{2}\Bigr) < 0.207,          \label{eq:hUpperBoundOdd} \\
\frac{1-\sqrt{1-h(n-3)}}{2h(n-3)} < \frac{1-\sqrt{1-0.207}}{2\times0.207}<0.265.  \label{eq:(1-sqrt)UpperBoundOdd}
\end{gather}
The first inequality in \eref{eq:hUpperBoundOdd} follows from \eref{eq:hnbound2};
the first inequality in \eref{eq:(1-sqrt)UpperBoundOdd} follows from \eref{eq:hUpperBoundOdd} and the fact that the real-valued function $(1-\sqrt{1-x})/(2x)$ is monotonically increasing in $x$ when $0<x\leq1$.
Therefore, $\sqrt{n}\nu(\Omega_{W_n})<3/8$ when $n$ is odd and $n\geq3, n\ne5$, which implies the upper bound in \eref{eq:bound-nuW}.

\item[2.] $n$ is an even integer.

Direct calculation based on \eqsref{eq:hn}{eq:nuW} shows that  $\sqrt{n}\nu(\Omega_{W_n})<0.31$ for $4\leq n\leq52$. When $n\geq 54$, we have
\begin{equation}
\sqrt{n}\,\nu(\Omega_{W_n})
=\sqrt{n-3}\,h(n-3) \sqrt{\frac{n}{n-3}} \frac{1-\sqrt{1-h(n-3)}}{2h(n-3)}
< \sqrt{\frac{\pi}{2}}  \tanh\Bigl(\frac{\pi}{2}\Bigr)  \times  1.03  \times   0.261  < 0.31,
\end{equation}
where the first inequality follows from \eref{eq:hnbound1} and the following equations,
\begin{gather}
\sqrt{\frac{n}{n-3}}\leq \sqrt{\frac{54}{54-3}}<1.03,  \\
h(n-3)\leq \sqrt{\frac{\pi}{2(n-3)}}  \tanh\Bigl(\frac{\pi}{2}\Bigr)
\leq \sqrt{\frac{\pi}{2(54-3)}}  \tanh\Bigl(\frac{\pi}{2}\Bigr) < 0.161,          \label{eq:hUpperBoundEven}  \\
\frac{1-\sqrt{1-h(n-3)}}{2h(n-3)} < \frac{1-\sqrt{1-0.161}}{2\times0.161}<0.261.  \label{eq:(1-sqrt)UpperBoundEven}
\end{gather}
The first inequality in \eref{eq:hUpperBoundEven} follows from \eref{eq:hnbound1};
the first inequality in \eref{eq:(1-sqrt)UpperBoundEven} follows from \eref{eq:hUpperBoundEven} and the fact that the real-valued function $(1-\sqrt{1-x})/(2x)$ is monotonically increasing in $x$ when $0<x\leq1$.
Therefore, $\sqrt{n}\nu(\Omega_{W_n})<0.31$ when $n$ is even and $n\geq4$,  which implies  the upper bound in \eref{eq:bound-nuW} again.
\end{enumerate}
In conclusion, \eref{eq:bound-nuW} holds for any integer $n$ that satisfies $n\geq 3$.

\section{\label{app:Proof Lim sqrt(n)nu}Proofs of \eqsref{eq:Lim sqrt(2n+1)nu}{eq:Lim sqrt(2n)nu}}
\begin{proof}
Equation \eqref{eq:Lim sqrt(2n+1)nu} can be proved as follows,
\begin{align}
&\lim_{n \to +\infty} \sqrt{2n+1}\nu(\Omega_{W_{2n+1}})
=\lim_{n \to +\infty} \sqrt{2n+1}\,\frac{1-\sqrt{1-h(2n-2)}}{2}  \nonumber\\
&=\left[\lim_{n \to +\infty} \sqrt{2n-2}\,h(2n-2)     \right]
\left(\lim_{n \to +\infty} \sqrt{\frac{2n+1}{2n-2}}  \right)
\left[\lim_{n \to +\infty} \frac{1-\sqrt{1-h(2n-2)}}{2h(2n-2)}\right]
= \frac{1}{4} \sqrt{\frac{\pi}{2}}  \coth\Bigl(\frac{\pi}{2}\Bigr)
\approx 0.342.
\end{align}
Here the first equality follows from \eref{eq:nuW}; the third one follows from \eref{eq:LimSqrt(n)h(n)Even} and the fact that $\lim_{n \to +\infty}h(n)=0$.

Equation \eqref{eq:Lim sqrt(2n)nu} can be proved as follows,
\begin{align}
&\lim_{n \to +\infty} \sqrt{2n}\nu(\Omega_{W_{2n}})
= \lim_{n \to +\infty} \sqrt{2n}\,\frac{1-\sqrt{1-h(2n-3)}}{2}  \nonumber\\
&=\left[\lim_{n \to +\infty} \sqrt{2n-3}\,h(2n-3)     \right]
\left(\lim_{n \to +\infty} \sqrt{\frac{2n}{2n-3}}  \right)
\left[\lim_{n \to +\infty} \frac{1-\sqrt{1-h(2n-3)}}{2h(2n-3)}\right]
= \frac{1}{4} \sqrt{\frac{\pi}{2}} \tanh\Bigl(\frac{\pi}{2}\Bigr)
\approx 0.287 .
\end{align}
The first equality follows from \eref{eq:nuW}; the third one follows from \eref{eq:LimSqrt(n)h(n)Odd} and the fact that $\lim_{n \to +\infty}h(n)=0$.	
\end{proof}

\section{\label{app:Wsymmetrization}Proofs of \eqssref{eq:TraceP1P2}{eq:lambda(Omega^G)}{eq:nu(Omega^G)Bound}}
\begin{proof}[Proof of \eref{eq:TraceP1P2}]
The equality $\tr(P_1 P_2^G)=\tr(P_1 P_2)$ in \eref{eq:TraceP1P2}	follows from the fact that $P_1$ is invariant under the action of $G$, that is, $P_1^G=P_1$.
The second equality in \eref{eq:TraceP1P2}
can be derived from   Eqs.~(\ref{eq:P1}) and~(\ref{eq:P2}) as follows,
\begin{align}
\tr(P_1 P_2)&=\<00\dots01|P_2|00\dots01\>+  \sum_{u\in B_{n-1}^1}(\<u|\otimes\<0|)P_2(|u\>\otimes|0\>)  \nonumber\\
&=h(n-1)+(n-1)[1-h(n-1)]
 =n-1-(n-2)h(n-1),
\end{align}
where  $B_{n-1}^1$ is the set of strings in $\{0,1\}^{n-1}$ with Hamming weight 1.  Here the second equality follows from the following equations,
\begin{align}
&\<00\dots01|P_2|00\dots01\>
=\sum_{x\in\{0,1\}^{n-1}}|\<00\dots0|\alpha_{x}\>|^2 \cdot |\<1|\beta_{x}\>|^2     =\frac{1}{2^{n-1}}\sum_{x\in\{0,1\}^{n-1}}\frac{1}{1+(n-1-2|x|)^2}
  \nonumber\\
&\quad=\frac{1}{2^{n-1}}\sum_{j=0}^{n-1}\frac{\binom{n-1}{j}}{1+(n-1-2j)^2}
=h(n-1), \\
&(\<u|\otimes\<0|)P_2(|u\>\otimes|0\>)
=\sum_{x\in\{0,1\}^{n-1}}|\<u|\alpha_{x}\>|^2 \cdot |\<0|\beta_{x}\>|^2              =\frac{1}{2^{n-1}}\sum_{x\in\{0,1\}^{n-1}}\left[1-\frac{1}{1+(n-1-2|x|)^2}\right]
 \nonumber\\
&\quad=1 - \frac{1}{2^{n-1}}\sum_{j=0}^{n-1}\frac{\binom{n-1}{j}}{1+(n-1-2j)^2}
=1-h(n-1),\qquad u\in B_{n-1}^1.
\end{align}

\end{proof}
\begin{proof}[Proof of \eref{eq:lambda(Omega^G)}]
Note that $\Omega_{W_{n}}^G$ can be expressed as follows,
\begin{equation}
\Omega_{W_{n}}^G=pP_1+(1-p)P_2^G=pP_1+(1-p)P_1P_2^GP_1+(1-p)(\openone-P_1)P_2^G(\openone-P_1),
\end{equation}	
given that $P_1$ and $P_2^G$ commute with each other. Therefore,
\begin{align}
\lambda_2(\Omega_{W_{n}}^G)&=\max\Bigl\{ p+(1-p)\bigl\|P_1 \bar{P}_2^G P_1\bigr\|,  (1-p)\bigl\|(\openone-P_1)P_2^G(\openone-P_1)\bigr\| \Bigr\}
=1-p=\frac{n-1}{n+(n-2)h(n-1)}.
\end{align}
Here  the second equality follows from the equality   $p+(1-p)\|P_1 \bar{P}_2^G P_1\|=1-p$ [cf.~\eref{eq:OptProbW}] and the inequality $(1-p)\|(\openone-P_1)P_2^G(\openone-P_1)\|\leq 1-p$.
\end{proof}

\begin{proof}[Proof of \eref{eq:nu(Omega^G)Bound}]
The equalities in \eref{eq:nu(Omega^G)Bound} follow from \eref{eq:lambda(Omega^G)}.

When  $3\leq n\leq40$, the lower bound in \eref{eq:nu(Omega^G)Bound} can be verified directly by virtue of \eref{eq:hn}. When  $n\geq41$, the lower bound  follows from the following equation
\begin{equation}\label{eq:sqrt{n}}
\sqrt{n}+(n-2)\sqrt{n} h(n-1)+1> \sqrt{n}+n-1> n.
\end{equation}
Here the first inequality is a consequence of  the inequalities $\sqrt{n} h(n-1)> \sqrt{n-1} h(n-1)>1$, the second of which follows from  Proposition~\ref{pro:sqrt{n}h(n)Monot} and the assumption $n\geq41$, given that
$\sqrt{40}\,h(40)>\sqrt{41}\,h(41)>1$.
This observation completes the proof of \eref{eq:nu(Omega^G)Bound}.
\end{proof}

\section{\label{app:TheoDstateGenProof}Proof of \tref{thm:DstateGen}}
\begin{proof}
In analogy to $\Omega_\bfk$ (cf. Appendix~\ref{app:TheoDickeProof}),
the verification operator $\Omega_\bfk^\phi$ can be expressed as
\begin{align}\label{eq:decomOmegaD'}
\Omega_\bfk^\phi
&= \binom{n}{2}^{-1}\sum_{i<j} \sum_{s=0}^g \bcaZ_{i,j}(\bfk_{ss})       \otimes \left[(\ket{s}\bra{s})^{\otimes2}\right]_{i,j}
      +\binom{n}{2}^{-1}\sum_{i<j}\sum_{s<t}\sum_{u\in B(\bfk_{st})}|u\>\<u| \otimes
                \big(\Gamma_{i,j,u}^+\otimes \Gamma_{j,i,u}^+ + \Gamma_{i,j,u}^- \otimes \Gamma_{j,i,u}^-\big)_{i,j}\nonumber\\
&=\frac{1}{n(n-1)} \bigg(\sum_{s=0}^r k_s^2-n\bigg) \mathcal{Z}(\bfk)
      +\frac{2}{n(n-1)}\sum_{i<j}\sum_{s<t}   \bcaZ_{i,j}(\bfk_{st})
      \otimes \Big[\bigl(\ket{\varphi_{s,t}^+}\bra{\varphi_{s,t}^+}\bigr)_{i,j} + \frac{1}{2}\bigl(|st\>\<st|+|ts\>\<ts|\bigl)_{i,j}\Big]  \nonumber\\
      &\ \quad +\frac{1}{n(n-1)}\sum_{i<j}\sum_{s<t}\sum_{u\in B(\bfk_{st})}|u\>\<u| \otimes
                \big[\rme^{\rmi\phi(v(i,j,u))}|st\>\<ts|\rme^{-\rmi\phi(v(j,i,u))}
                    +\rme^{\rmi\phi(v(j,i,u))}|ts\>\<st|\rme^{-\rmi\phi(v(i,j,u))} \big]_{i,j} \nonumber\\
&= \frac{1}{n(n-1)}\bigg[M'_1+\sum_{s<t}M_{(s,t)}\bigg]\,.
\end{align}
Here $\ket{\varphi_{s,t}^+}=\frac{1}{\sqrt{2}}(\ket{s}\ket{s}+\ket{t}\ket{t})$,
$\caZ(\bfk)=\sum_{u\in B(\bfk)}\ket{u}\bra{u}$,
$\bcaZ_{i,j}(\bfk_{st})$ is defined in \eref{eq:barZijst}, $v(i,j,u)$ and $v(j,i,u)$ are defined in \eqsref{eq:viju}{eq:vjiu},
$M_{(s,t)}$ is defined in \eref{eq:M(s,t)}, and
\begin{align}
M'_1:=&\;\bigg(\sum_{s=0}^r k_s^2-n+\sum_{s<t} k_s k_t \bigg)\caZ(\bfk)
     +\sum_{\substack{u,v\in B(\bfk)\\ u\sim v }} \big(\rme^{\rmi\phi(u)}\ket{u}\big)\big(\bra{v}\rme^{-\rmi\phi(v)}\big) \nonumber\\
    =&\;\frac{1}{2}\bigg(n^2-2n+\sum_{s=0}^r k_s^2 \bigg)\caZ(\bfk)
     +\sum_{u,v\in B(\bfk) }A_{u v}\big(\rme^{\rmi\phi(u)}\ket{u}\big)\big(\bra{v}\rme^{-\rmi\phi(v)}\big)\,,
\end{align}
where the notation $u\sim v$ means  $u_j\neq v_j$ for exactly two values of $j$.
The coefficient matrix $(A_{uv})$ for $u,v\in B(\bfk)$ happens to be the
adjacency matrix $A(\bfk)$ of the transposition graph $G(\bfk)$ \cite{Chase1973} (cf. Appendix~\ref{app:SpecGraph}).

Note that $M'_1$ can be turned into $M_1$ in \eref{eq:M1} by a diagonal unitary transformation; similarly, $\Omega_\bfk^\phi$ can be turned into $\Omega_\bfk$   by a diagonal unitary transformation [cf.~Appendix~\ref{app:TheoDickeProof}]. Therefore, $\Omega_\bfk^\phi$ and  $\Omega_\bfk$ have the same spectrum and the same spectral gap. Thanks to \eqsref{eq:gapD}{eq:nuOmegaD}, we have
\begin{equation}
\nu\big(\Omega_\bfk^\phi\big)=\nu\big(\Omega_\bfk\big)=
\min\left\{\frac{1}{n-1},1-\frac{ k_0(k_0+1)+k_1(k_1+1) }{2n(n-1)} \right\}=\begin{cases}
1/2          & \bfk=(1,1,1),  \\
1/3          & \bfk=(2,1),    \\
1/(n-1) \ \  & n\ge4.
\end{cases}
\end{equation}
This result  confirms \eref{eq:nuOmegaD'} and implies \eref{eq:NOmegaD'} in view of \eref{eq:NumberTest} (cf.~Theorem~\ref{thm:Dicke}).
\end{proof}

\section{\label{app:TheoAntiStateProof}Proof of \tref{thm:Antisymmetric}}
\begin{proof}
	The verification operator $\Omega_{\AS_n}$ can be expressed as
	\begin{align}
	\Omega_{\AS_n}
	&= \binom{n}{2}^{-1}\sum_{i<j}\sum_{s<t}\bcaZ_{i,j}(\bfk_{st}) \otimes \big(T_{s,t}^+\otimes T_{s,t}^- + T_{s,t}^-\otimes T_{s,t}^+\big)_{i,j}\nonumber\\
	&= \frac{2}{n(n-1)}\sum_{i<j}\sum_{s<t}   \bcaZ_{i,j}(\bfk_{st})
	\otimes \Big[\bigl(\ket{\psi_{s,t}^-}\bra{\psi_{s,t}^-}\bigr)_{i,j}+\bigl(\ket{\varphi_{s,t}^-}\bra{\varphi_{s,t}^-}\bigr)_{i,j}\Big] \nonumber\\
	&= \frac{1}{n(n-1)}\bigg[M_1^{\AS}+\sum_{s<t}\sum_{\substack{u\in B({\bfk}^s_{t})\\v\in B({\bfk}^t_{s})\\ u\sim v}}X_{s,t}^{u,v}\bigg]\,.\label{eq:decomOmegaAS}
	\end{align}
	Here $\ket{\psi_{s,t}^-}=\frac{1}{\sqrt{2}}(\ket{s}\ket{t}-\ket{t}\ket{s})$, $\ket{\varphi_{s,t}^-}=\frac{1}{\sqrt{2}}(\ket{s}\ket{s}-\ket{t}\ket{t})$,
	\begin{align}
	M_1^{\AS}&:=\frac{n(n-1)}{2} \caZ(\bfk)-\sum_{\substack{u,v\in B(\bfk)\\ u\sim v }}\ket{u}\bra{v}
	=\frac{n(n-1)}{2} \caZ(\bfk)-\sum_{u,v\in B(\bfk) }A_{u v}\ket{u}\bra{v}\,,\label{eq:M1AS}\\
	X_{s,t}^{u,v}&:=\ket{u}\bra{u}+\ket{v}\bra{v}-\ket{u}\bra{v}-\ket{v}\bra{u},
	\end{align}
and the notation $u\sim v$ means  $u_j\neq v_j$ for exactly two values of $j$.	
In addition,  the coefficient matrix $(A_{uv})$ for $u,v\in B(\bfk)$ happens to be the
	adjacency matrix $A(\bfk)$ of the transposition graph $G(\bfk)$ \cite{Chase1973}.
	Note that $M_1^{\AS}$ and all $X_{s,t}^{u,v}$  [with $s<t$, $u\in B\big({\bfk}^s_{t}\big)$, $v\in B\big({\bfk}^t_{s}\big)$, and $ u\sim v$]  are Hermitian and have mutually orthogonal supports, so all of them are
	positive semidefinite given that $\Omega_{\AS_n}$ is positive semidefinite by construction.

	According to \lref{lem:CaylGraphSpect} in Appendix~\ref{app:SpecGraph}, the smallest eigenvalue of $A(\bfk)$ is equal to $-n(n-1)/2$ with multiplicity~1,
	and the second smallest eigenvalue of $A(\bfk)$ is $n-n(n-1)/2$.
	Therefore, the two largest eigenvalues of $M_1^{\AS}$ read
	\begin{equation}\label{eq:M1ASlambda}
	\lambda_1(M_1^{\AS})=n(n-1), \qquad \lambda_2(M_1^{\AS})=n(n-1)-n=n(n-2).
	\end{equation}
In addition, direct calculations show that the maximum eigenvalue of $X_{s,t}^{u,v}$ is 2. In conjunction with Eqs.~\eqref{eq:decomOmegaAS} and \eqref{eq:M1ASlambda}, we can deduce
	the second largest eigenvalue and  spectral gap of $\Omega_{\AS_n}$, with the result (assuming $n\geq3$)
	\begin{align}
	\lambda_2\big(\Omega_{\AS_n}\big)&=\max\left\{\frac{\lambda_2(M_1^{\AS})}{n(n-1)},\,\max_{s<t}\frac{\lambda_1(X_{s,t}^{u,v})}{n(n-1)}\right\}=\frac{n-2}{n-1},\\
	\nu\big(\Omega_{\AS_n}\big)&=1-\lambda_2\big(\Omega_{\AS_n}\big)=\frac{1}{n-1},
	\label{eq:gapAS}
	\end{align}
	which confirms \eref{eq:nuOmegaAS}.
	
	Equation \eqref{eq:NumberTestAS} follows from \eqsref{eq:NumberTest}{eq:nuOmegaAS}.
\end{proof}

\section{\label{app:ComputeOmegaASG}Proofs of Theorem~\ref{thm:OmegaASG} and \lref{lem:DimRatio}}
\begin{proof}[Proof of Theorem~\ref{thm:OmegaASG}]
	According to \eref{eq:OmegaSymProj},  we have
	\begin{equation}
	\Omega_{\AS_n}^{\tilde{G}}=\sum_{\mu\vdash n}\frac{1}{d_\mu D_\mu} \tr\bigl(\Omega_{\AS_n}^{\tilde{G}} P_\mu\bigr) P_\mu=\sum_{\mu\vdash n}\frac{1}{d_\mu D_\mu} \tr\bigl(\Omega_{\AS_n} P_\mu\bigr) P_\mu.
	\end{equation}
	By representation theory, the projector $P_\mu$ can be expressed as follows,
	\begin{equation}
	P_\mu =\frac{d_\mu}{n!}\sum_{\sigma\in \scrS_n}\chi_\mu(\sigma)  U_\sigma,
	\end{equation}	
	where $U_\sigma$ is the unitary operator  corresponding to the permutation $\sigma$, and  $\chi_\mu(\sigma)$ is the character of $\sigma$ associated with the representation labeled by $\mu$. Therefore,
	\begin{equation}
	\tr\bigl(\Omega_{\AS_n} P_\mu\bigr)=\frac{1}{n(n-1)}\bigg[\tr\bigl(P_\mu M_1^{\AS}\bigr)+\sum_{s<t}\sum_{\substack{u\in B({\bfk}^s_{t})\\v\in B({\bfk}^t_{s})\\ u\sim v}}\tr\bigl(P_\mu X_{s,t}^{u,v}\bigr)\bigg]=d_\mu^2,
	\end{equation}
	which implies \eref{eq:OmegaASG}.
	Here the first equality follows from \eref{eq:decomOmegaAS}, and the notation $u\sim v$ means  $u_j\neq v_j$ for exactly two values of $j$. The second equality follows from \eqsref{eq:PmuM}{eq:PmuX} below,
	\begin{align}
	&\tr\bigl(P_\mu M_1^{\AS}\bigr)=\frac{n(n-1)}{2}\tr[P_\mu \caZ(\bfk) ]-\sum_{u,v\in B(\bfk)} A_{u,v} \<v|P_\mu |u\>
	=\frac{n(n-1)}{2}\sum_{u\in B(\bfk)} \<u |P_\mu |u\> -\sum_{\substack{u,v\in B(\bfk)\\ u\sim v }} \<v|P_\mu |u\>\nonumber\\
	&\quad =\frac{n(n-1)}{2( n!)}d_\mu^2 |B(\bfk)|
	- \sum_{\substack{u,v\in B(\bfk)\\ u\sim v }}  \frac{d_\mu \chi_\mu(\tau)}{n!}=\frac{n(n-1)}{2}d_\mu^2  -\frac{n(n-1)}{2}d_\mu \chi_\mu(\tau), \label{eq:PmuM}\\
	&\sum_{s<t}\sum_{\substack{u\in B({\bfk}^s_{t})\\v\in B({\bfk}^t_{s})\\ u\sim v}}\tr\bigl(P_\mu X_{s,t}^{u,v}\bigr)=\sum_{s<t}\sum_{\substack{u\in B({\bfk}^s_{t})\\v\in B({\bfk}^t_{s})\\ u\sim v}}(\< u |P_\mu |u\> +\<v | P_\mu |v\>)=2\sum_{s<t}\sum_{u\in B({\bfk}^s_{t})}\< u |P_\mu |u\> \nonumber\\
	&\quad =n(n-1)|B({\bfk}^s_{t})|\frac{d_\mu}{n!}[d_\mu +\chi_\mu(\tau)]
	=\frac{n(n-1)}{2}d_\mu^2 +\frac{n(n-1)}{2}d_\mu \chi_\mu(\tau), \label{eq:PmuX}
	\end{align}
	where $\tau\in \scrS_n$ is any transposition.
	
	Alternatively, the trace $\tr(\Omega_{\AS_n} P_\mu)$ can be derived by virtue of \eqsref{eq:PijAS}{eq:OmegaAS} as follows,
	\begin{align}
	&\tr\bigl(\Omega_{\AS_n} P_\mu\bigr)
	=\binom{n}{2}^{-1} \sum_{i<j} \tr\bigl(P_\mu P_{i,j}^{\AS}\bigr)
	=\tr\bigl(P_\mu P_{1,2}^{\AS}\bigr)=\sum_{s<t}\tr\Bigl\{P_\mu \bigl[\left(T_{s,t}^+\otimes T_{s,t}^- + T_{s,t}^-\otimes T_{s,t}^+\right)\otimes \bcaZ(\bfk_{st})\bigr]\Bigr\}\nonumber\\
	&=\sum_{s<t}\tr\Bigl\{P_\mu U_{st}^{\otimes n} \bigl[\left(T_{s,t}^+\otimes T_{s,t}^- + T_{s,t}^-\otimes T_{s,t}^+\right)\otimes \bcaZ(\bfk_{st})\bigr]U_{st}^{\otimes n\dag}\Bigr\}=\sum_{s<t}\tr\Bigl\{P_\mu  \bigl[\left(|st\>\<st|+|ts\>\<ts|\right)\otimes \bcaZ(\bfk_{st})\bigr]\Bigr\}\nonumber\\
	&
	=\tr[P_\mu \caZ(\bfk)]=\sum_{u\in B(\bfk)} \<u|P_\mu |u\>
	=\sum_{u\in B(\bfk)}\frac{d_\mu}{n!}\sum_{\sigma\in \scrS_n} \chi_\mu(\sigma)\<u|U_\sigma|u\>
	=\frac{d_\mu^2}{n!}|B(\bfk)|=d_\mu^2,
	\end{align}
	where
	\begin{equation}
	U_{st}=\frac{1}{\sqrt{2}}(|s\>\<s|+|s\>\<t| +|t\>\<s|-|t\>\<t|)+\sum_{r\neq s,t}|r\>\<r|.
	\end{equation}
	
	Equation \eqref{eq:nuOmegaASG} in Theorem~\ref{thm:OmegaASG} follows from \eref{eq:OmegaASG} and \lref{lem:DimRatio}.
	Equation \eqref{eq:NumberTestASG} follows from \eqsref{eq:NumberTest}{eq:nuOmegaASG}.
\end{proof}

\begin{proof}[Proof of \lref{lem:DimRatio}]
	According to the well-known dimension formulas for $D_\mu$ and $d_\mu$ (see Refs.~\cite{Weyl1931,Proc2007} for example), we have
	\begin{equation}
	\frac{D_\mu }{d_\mu}=\frac{1}{n!}\prod_{j=1}^d \frac{(d+\mu_j-j)!}{(d-j)!}=\frac{1}{n!}\prod \frac{\Gamma(d+\mu_j-j+1)}{\Gamma(d-j+1)},
	\end{equation}
	where $d=n$ is the local dimension. Note that this formula is still applicable when $d\neq n$. As an implication, we have
	\begin{align}
	\ln \frac{D_\mu }{d_\mu}=\sum_{j=1}^d \ln \Gamma(d+\mu_j-j+1)-\sum_{j=1}^d\ln \Gamma(d-j+1)-\ln (n!).
	\end{align}
	Recall that  the function $\ln\Gamma(x)$ is convex in $x$ for $x>0$, we conclude that 	$\ln \frac{D_\mu }{d_\mu}$ is convex and thus Schur convex in $\mu$. Therefore, $\ln \frac{D_\mu }{d_\mu}\leq \ln \frac{D_{\mu'}}{d_{\mu'}}$ whenever $\mu \prec\mu'$, which confirms \lref{lem:DimRatio}.
\end{proof}

\section{\label{app:SpecGraph}The spectrum of the transposition graph}
Let $\bfk:=(k_0,k_1,\dots,k_r)$ with $k_0,\dots,k_r$ being positive integers and $n=\sum_{j=0}^{r}k_j$. Recall that  $B(\bfk)$ is the set of all sequences of $n$ symbols in which $k_s$ symbols are equal to $s$ for $s=0,1,\dots,r$.
The transposition graph $G(\bfk)$ is a regular graph whose vertices are labeled by sequences in $B(\bfk)$.
Two distinct vertices $u,v\in B(\bfk)$ are adjacent iff $u$ and $v$ can be turned into each other by a transposition \cite{Chase1973}, that is, $u_j\neq v_j$ for exactly two values of $j$.
The number of vertices in $G(\bfk)$ is equal to the cardinality of $B(\bfk)$, which reads $|B(\bfk)|=n!/\big(\prod _{j=0}^r k_j!\big)$, and the degree of $G(\bfk)$ is given by
\begin{equation}
d:=\frac{1}{2}\bigg(n^2-\sum_{s=0}^r k_s^2\bigg).
\end{equation}
Let $A(\bfk)$ be the adjacency matrix of $G(\bfk)$. The eigenvalues of $G(\bfk)$ are defined
as the eigenvalues of $A(\bfk)$.
Here we are interested in the largest and second largest  eigenvalues of $G(\bfk)$, which are crucial to the proof of  Theorems~\ref{thm:Dicke} and~\ref{thm:DstateGen}. The  lemma below  follows from Eq.~(4.2) in Ref.~\cite{Caputo2010}.
\begin{lemma}\label{lem:GraphSpect}
	The largest eigenvalue of $G(\bfk)$ is equal to its degree $d$ and has multiplicity 1;
	the second largest eigenvalue of $G(\bfk)$ is equal to $d-n$.
\end{lemma}

When $k_0=k_1=\cdots k_{n-1}=1$, the graph $G(\bfk)$ reduces to the Cayley graph of the symmetric group. In this case we can also determine the smallest and second smallest eigenvalues of $G(\bfk)$. To this end, note that the sequences in $B(\bfk)$ can be divided into two groups of equal size: one group can be constructed from the sequence $(0,1,\ldots, n-1)$ by even permutations, and the other group can be constructed by odd permutations. In addition, $G(\bfk)$ is a bipartite graph with respect to this partition; accordingly, the adjacency matrix $A(\bfk)$ has a block form,
\begin{equation}
A=\begin{pmatrix}
0 &B\\
B^T & 0
\end{pmatrix},
\end{equation}
where $B$ is a matrix of size $(n!/2)\times (n!/2)$. Therefore, the eigenvalues of $A(\bfk)$ form pairs: if $\lambda$ is an eigenvalue of $A$, then $-\lambda$ is an eigenvalue with the same multiplicity. Together with \lref{lem:GraphSpect}, this observation implies the following lemma (cf. Aldous' spectral gap
conjecture, which was proved in Ref.~\cite{Caputo2010}).
\begin{lemma}\label{lem:CaylGraphSpect}
	Suppose $\bfk=(k_0,k_1,\dots,k_{n-1})$ with $k_j=1$ for $j=0,1,\dots,n-1$.
	Then the smallest eigenvalue of $G(\bfk)$ is equal to $-d$ and has multiplicity 1, where $d=n(n-1)/2$ is the degree of $G(\bfk)$;
	the second smallest eigenvalue of $G(\bfk)$ is equal to $n-d$.
\end{lemma}

Let us take $\bfk=(1,1,1)$ for example. In this case we have $n=d=3$, and  $G(\bfk)$ is a bipartite graph with six vertices
labeled by the sequences $(0,1,2)$, $(1,2,0)$, $(2,0,1)$,  $(0,2,1)$, $(1,0,2)$,  $(2,1,0)$. With respect to this order,  the adjacency matrix of $G(\bfk)$ reads
\begin{equation}
A(\bfk)=\begin{pmatrix}
0&0&0&1&1&1\\
0&0&0&1&1&1\\
0&0&0&1&1&1\\
1&1&1&0&0&0\\
1&1&1&0&0&0\\
1&1&1&0&0&0\\
\end{pmatrix}.
\end{equation}
Direct calculation shows that $A(\bfk)$ has three distinct eigenvalues $3,0,-3$, with multiplicities $1,4,1$, respectively, which agrees with Lemmas~\ref{lem:GraphSpect} and \ref{lem:CaylGraphSpect}.

\twocolumngrid

\end{document}